\documentclass[sigconf, nonacm]{acmart}

\usepackage{amsmath,amsfonts}
\usepackage{algorithmic}
\usepackage{graphicx}
\usepackage{textcomp}
\usepackage{xcolor}
\usepackage{url}
\usepackage{xurl}
\usepackage{pdfpages}
\def\BibTeX{{\rm B\kern-.05em{\sc i\kern-.025em b}\kern-.08em
    T\kern-.1667em\lower.7ex\hbox{E}\kern-.125emX}}

\usepackage{booktabs}
\usepackage{xspace}
\usepackage{fancyvrb}
\usepackage{graphicx}
\usepackage{tikz}
\usepackage{bm}
\usepackage{subfigure}
\usepackage{paralist}

\usepackage{weiwMath}
\usepackage{weiwAlgorithm}
\usepackage{weiwFixme}
\usepackage{weiwThm}
\usepackage{color}

\newcommand{\MODELNAME}{LMSFC\xspace}
\newcommand{\bplustree}{$B^+$-tree\xspace}

\newcommand{\md}{multi-dimen\-sional\xspace}
\newcommand{\Md}{Multi-dimen\-sional\xspace}
\newcommand{\oned}{one-dimensional\xspace}
\newcommand{\zaddress}{$z$-address\xspace}

\newcommand{\qmin}{q_{\min}}
\newcommand{\qmax}{q_{\max}}

\newcommand{\SMBO}{SMBO\xspace}

\newcommand{\coord}[2]{#1^{(#2)}} 

\newtheorem{definition}{Definition}

\newcommand{\myparagraph}[1]{\noindent\textbf{#1}\hspace{1em}}

\newcommand{\ZM}{\textsf{ZM-index}\xspace}
\newcommand{\zm}{\ZM}
\newcommand{\rstartree}{\textsf{$R^{*}$-tree}\xspace}
\newcommand{\rtree}{$R$-tree\xspace}
\newcommand{\Flood}{\textsf{Flood}\xspace}

\newcommand{\RSMI}{\textsf{RSMI}\xspace}
\newcommand{\LISA}{\textsf{LISA}\xspace}
\newcommand{\TSUNAMI}{\textsf{Tsunami}\xspace}

\newcommand{\RQS}{\textsf{RQS}\xspace}
\newcommand{\FNZ}{\textsf{FNZ}\xspace}

\newcommand{\zorder}{$z$-order\xspace}

\newcommand{\mycomment}[1]{}

\newcommand{\opt}[1]{#1^{*}}

\newtheorem{example}{Example}

\newcommand\vldbdoi{10.14778/3603581.3603598}
\newcommand\vldbpages{2605 - 2617}
\newcommand\vldbvolume{16}
\newcommand\vldbissue{10}
\newcommand\vldbyear{2023}
\newcommand\vldbauthors{Jian Gao, Xin Cao, Xin Yao, Gong Zhang and Wei Wang}
\newcommand\vldbtitle{\shorttitle} 
\newcommand\vldbavailabilityurl{URL_TO_YOUR_ARTIFACTS}
\newcommand\vldbpagestyle{plain} 

\begin{document}

\title{{\MODELNAME}: A Novel Multidimensional Index based on Learned Monotonic Space Filling Curves (Extended Version)}

\author{Jian Gao}
\affiliation{%
  \institution{UNSW, Australia}
}
\email{jian.gao2@unsw.edu.au}

\author{Xin Cao}
\affiliation{%
  \institution{UNSW, Australia}
}
\email{xin.cao@unsw.edu.au}

\author{Xin Yao}
\affiliation{%
  \institution{Huawei Theory Lab, China}
}
\email{yao.xin1@huawei.com}

\author{Gong Zhang}
\affiliation{%
  \institution{Huawei Theory Lab, China}
}
\email{nicholas.zhang@huawei.com}

\author{Wei Wang$^{1,2}$}
\affiliation{%
  \institution{$^{1}$DSA \& Guangzhou Municipal Key Laboratory of Materials Informatics, HKUST (Guangzhou), China}
  \institution{$^{2}$HKUST, HKSAR, China}
}
\email{weiwcs@ust.hk}

\begin{abstract}

The recently proposed learned indexes have attracted much attention as they
  can adapt to the actual data and query distributions to attain better search 
  efficiency. Based on this technique, several existing works build up indexes for \md{} data and achieve improved query performance. A common paradigm of these works is to (i) map \md{} data points to a one-dimensional space using a fixed space-filling curve (SFC) or its variant and (ii) then apply the learned indexing techniques. We notice that the first step typically uses a fixed SFC method, such as row-major order and \zorder{}.
  It definitely limits the potential of learned \md{} indexes to adapt variable data distributions via different query workloads.

  In this paper, we propose a novel idea of learning a space-filling curve
  that is carefully designed and actively optimized for efficient query processing.
  We also identify innovative offline and online optimization opportunities common to SFC-based learned indexes and offer optimal and/or heuristic solutions. 
  Experimental results demonstrate that our proposed method, \MODELNAME, outperforms state-of-the-art non-learned or learned methods across three commonly used real-world datasets and diverse experimental settings.

\end{abstract}

\maketitle

\pagestyle{\vldbpagestyle}
\begingroup\small\noindent\raggedright\textbf{PVLDB Reference Format:}\\
\vldbauthors. \vldbtitle. PVLDB, \vldbvolume(\vldbissue): \vldbpages, \vldbyear.\\
\href{https://doi.org/\vldbdoi}{doi:\vldbdoi}
\endgroup
\begingroup
\renewcommand\thefootnote{}\footnote{\noindent
This work is licensed under the Creative Commons BY-NC-ND 4.0 International License. Visit \url{https://creativecommons.org/licenses/by-nc-nd/4.0/} to view a copy of this license. For any use beyond those covered by this license, obtain permission by emailing \href{mailto:info@vldb.org}{info@vldb.org}. Copyright is held by the owner/author(s). Publication rights licensed to the VLDB Endowment. \\
\raggedright Proceedings of the VLDB Endowment, Vol. \vldbvolume, No. \vldbissue\ %
ISSN 2150-8097. \\
\href{https://doi.org/\vldbdoi}{doi:\vldbdoi} \\
}\addtocounter{footnote}{-1}\endgroup


\section{Introduction}
\label{sec:introduction}

Nowadays, there are large volumes and a huge variety of \md{} data. For example,
in traditional data warehouses and analytical databases, the majority of key
data is stored in the \emph{\md{}} fact table. The wide deployments of location-based services and sensors, such as Google
Maps, generate huge amounts of \emph{\md{}} data. The data is typically two or
three spatial dimensions, and one or several dimensions for various
measurements. 
The common and dominant type of query over these \md{} datasets is the window
query, which imposes range constraints on several or all the dimensions.

\Md{} indexes are essential in answering window queries efficiently for a large
volume of \md{} datasets. Previous studies have proposed many traditional
indexes, including \rtree~\cite{rtree}, kd-tree~\cite{kdtree}, and
Quadtree~\cite{quadtree}. They are all based on spatial partitioning, while the
major difference is whether the overlapping between partitions exists or not.

A space-filling curve is one of the most commonly used methods in \md{}
indexes~\cite{kamel1993hilbert, DBLP:conf/vldb/RamsakMFZEB00}. This is because
SFCs have excellent proximity-preserving properties, making them ideal for
linearizing data objects. SFCs can be classified into two categories: monotonic SFCs and non-monotonic
SFCs. Monotonic SFCs, such as \zorder~\cite{morton}, enable quick location of
the search range. However, non-monotonic SFCs may result in more computational
overhead during query processing. For instance, Hilbert curve
~\cite{hilbert1935stetige} requires the enumeration of all values on the
boundary of the query window to determine the search range. Therefore, it is
more difficult to perform range searches when using non-monotonic SFCs as the
linearization method.

Recently, initiated by the seminal work~\cite{rmi}, there is a surge in
optimizing database indexing via machine learning. It takes unique advantage of
optimizing for specific data and query workload instances. As a result, several
works have investigated the learned \md{} index. The prevalent approach is to
map the \md{} data points into one dimension, and further apply a learned index
on the one-dimensional space.

However, they have the following limitations: 
\begin{inparaenum}[(1)]
\item The unique and critical part of multi-dimensional indexes is mapping from
multi-dimensional space to one-dimensional space, while this is not learned or
well-learned. Most methods~\cite{ZM,RSMI} exploit an existing SFC, such as
\zorder~\cite{morton}, as it possesses good proximity-preserving capabilities.
However, a fixed SFC does not necessarily work the best on a given dataset
instance. Other works, such as \Flood~\cite{flood}, only learn to select a
special dimension and then follow the fixed row-major orders on the rest of the
$d-1$ dimensions, hence missing the opportunity to better preserve local
proximities. 
\item The physical layout of the linearized data points, which are stored as
disk pages, is not fully optimized. Existing methods typically pack a fixed
amount of data points into each page, which may cause the Minimum Bounding
Rectangles (MBRs) of the resulting pages to contain much dead space or heavily
overlap with each other~\cite{DBLP:conf/icde/SidlauskasCZA18}.
\item Even if the above two issues can be mitigated at index construction time,
existing query processing methods still need to visit many pages because their
MBRs or \zaddress{} ranges  \emph{inevitably} overlap with that of the query.
\end{inparaenum}

In this paper, we provide a thorough and rigorous investigation of learned \md{}
indexes and propose LMSFC to address the above limitations. 
\begin{inparaenum}[(i)]
\item It is a challenging task to formulate a suitable family of parameterized
  SFCs that can be efficiently learned and possess salient properties for
  efficient query processing. For this, we design a learnable monotonic SFC
  family. Based on the proposed SFC family, we devise an effective solution
  based on the SMBO to learn an optimal/sub-optimal SFC to adapt to different
  data distributions and workloads. \item We study the physical storage
  optimization issue of packing \md{} data points into size-limited disk pages
  in a principled fashion. 
  Thanks to the linearization due to the (learned) SFC, we are able to solve the
  \emph{otherwise} NP-hard problem optimally by dynamic programming. We further
  propose a heuristic algorithm that trades the optimality for improved
  practical speed, which is suitable for large-scale datasets. 
\item In addition to the above offline optimization techniques, we further
  exploit the unique online query optimization opportunity by proposing a query
  splitting strategy. We demonstrate that our learned SFC lends itself to an
  efficient algorithm of splitting the query into two such that the total access
  to false negative data pages is minimized; this algorithm is then extended to
  allow multiple splits via recursion.
\end{inparaenum}

The contributions of the paper are summarized below:

\begin{enumerate}[(1)]
 \item As far as we are aware, LMSFC is the \emph{first} work to consider
  learning a space-filling curve that is directly optimized for the least cost
  query processing on a given instance of the dataset and query workload.

\item Based on the learned SFC, we propose both offline and online optimization
  techniques. For the offline optimization, we can optimally or sub-optimally
  pack \md{} data points into pages that minimize a density based cost function.
  For the online optimization, we propose recursively splitting the query into
  sub-queries such that it minimizes the access to pages that do not contain any
  potential data points for a query. 

\item We compare LMSFC with previous state-of-the-art \md{} indexes in our
  experimental evaluation. LMSFC achieves the best query performance on three
  real-world datasets under varying query selectivity, data size, and query
  aspect ratio.

  \MODELNAME{} can achieve up to 38.2$\times$, 7.2$\times$, and 2.0$\times$
  speedup against \rstartree{}, \ZM{}, and \Flood{}, respectively. 

\end{enumerate}

The rest of this paper is organized as follows. %
In Section~\ref{sec:preliminaries}, we define the problem setting and introduce
some key notations and concepts. Then we review the related literature in
Section~\ref{sec:related}. Section~\ref{sec:fram-learn-sfcs} starts by
investigating the requirement of range query processing for an SFC-based index
and motivates a parameterized SFC family. We then overview the proposed
\MODELNAME{} index based on the parameterized SFC family and introduce key steps
in its construction in Section~\ref{sec:index-construction}. We motivate and
introduce our query processing method with the novel query splitting strategy in
Section~\ref{sec:query-processing}. Section~\ref{sec:exp} presents our
experimental results and analyses. Finally, we conclude the paper and discuss a
few extensions in Section~\ref{sec:conclusion}.

\section{preliminaries}
\label{sec:preliminaries}

\subsection{Problem Definition}

In this paper, we mainly focus on exact query processing of window queries on a
multi-dimensional dataset. 
\begin{definition}[Multi-dimensional Dataset]
  A \md{} dataset $D$ consists of $n$ points in a $d$-dimensional
  Euclidean space. Each point $x \in D$ can be denoted as $(x^{(1)},\cdots,x^{(d)})$
  where $x^{(i)} \in \mathbb{R}^d$ is the $i$-th dimensional value of $x$.
\end{definition}

Without loss of generality, we assume that, with proper scaling, the domain of
each coordinate is an integer within $[0, 2^{K} - 1]$, hence every dimension
value can be represented using $K$ binary bits.

A \md{} window is a hyper-rectangle in the $d$-di\-mensional 
space, or formally as $w = [x^{(1)}_{L}, x^{(1)}_{U}] \times \cdots \times [x^{(d)}_{L},
x^{(d)}_{U}]$, where $\coord{x}{i}_{L} \leq \coord{x}{i}_{U}$.

\begin{definition}[Multi-dimensional Window Query]
  Given a \md{} dataset $D$, a \md{} window query $q$ with the window constraint
  $q.w$ returns the set of points $x$ from $D$ that is located inside the window
  $q.w$,
  i.e.
  $R(q) = \set{x \mid x \in D \land \forall{1 \leq i \leq d}, \coord{q.w_L}{i} \leq
    \coord{x}{i} \leq \coord{q.w_U}{i}}$.
\end{definition}
As a query is uniquely characterized by its query window, we will use $q$ and
$q.w$ interchangeably hereafter.

\subsection{Notations}

Given a \md{} window $w$, it is uniquely characterized by $(w_L, w_U)$, i.e.,
the lower-bound and upper-bound points are defined as
$w_L = (x^{(1)}_{L}, \cdots, x^{(d)}_{L})$ and
$w_U = (x^{(1)}_{U}, \cdots, x^{(d)}_{U})$, respectively. These apply to the
window constraint of the query window of $q$ too, and we will use the shorthand $q_L$ to
denote $q.w_L$.

Given a set of points, we define the \md{}al \emph{Minimum Bounding Rectangle} (MBR) 
as the smallest window that encloses the set of points.

For a binary integer $v$, we denote the $j$-th bit of the binary representation
of $v$ as $v_j$. Note that $j$ starts from 0, which corresponds to the
right-most bit of $v$. The most significant bit of $v$ is the bit set to 1 and
with the maximum bit index. For example, let
$v = \mathtt{(00101101)_2}$ (we use $()_2$ to represents binary strings),
$v_2 = \mathtt{1}$ and the most significant bit of $v$ is 5.

Table~\ref{tab:notation} lists frequently used notations.

\begin{table}[tb]
  \caption{Table of Notations}
  \label{tab:notation}
  \small
  \begin{tabular*}{\linewidth}{@{\extracolsep{\fill} }  p{12mm} p{67mm} }
    \toprule
    Notation   &  Description\\
    \midrule
    $D$        &  A set of \md data records\\\hline
    $d$        &  The dimensionality of $D$ \\\hline
    $x$, $\coord{x}{i}_{j}$  &  A \md data point, and the $j$-th bit of the $i$-th dimension value of $x$ \\\hline
    $q$, $q_L, q_U$ &  A \md window query, and its lower-bounding and upper-bounding points, respectively\\\hline
    $f$        &  A learned SFC's mapping function \\\hline
    $\theta$   &  The parameters of $f$ \\\hline
    $K$        &  The maximum number of bits to represent coordinate values of $x$\\

    \bottomrule
  \end{tabular*}
\end{table}

\section{Related Work}
\label{sec:related}

Indexes are essential for processing queries on large datasets. Traditional
indexes are optimized for the worst-case performance. More importantly, they
miss the opportunity to exploit statistical information about the data and query
workloads to optimize their index structure and physical layout. Recently, many
Machine Learning-based indexes have been proposed, which achieve smaller index
sizes and/or faster query processing speed. %
RMI~\cite{rmi} is a well-known work that first notices the similarity between
the exact search on one-dimensional array of non-descending values and the
classic regression problem. It then proposes several instances of learned
indexes, which use a complex model to predict the logical location of the
targeted key and then perform a local search to fix possible errors bought by
the ML model. Later works further improve the model's performance or consider other variants.
For instance, PGM~\cite{pgm} lets the user specify the error bound a priori and
then uses simpler linear regression models with optimal segmentation algorithms
to construct the learned index. Fiting-tree~\cite{fittingtree} also uses a
tunable error parameter to tradeoff the index size for lookup performance. Given
a dataset, Fiting-tree applies a cost model to estimate the space consumption
and latency to find a suitable error bound. Radix
Spline~\cite{DBLP:conf/sigmod/KipfMRSKK020} considers a linear spline to
approximate CDF. The prefixes of the selected spline points are stored in an
auxiliary radix table to accelerate the search process. ALEX~\cite{ding2020alex}
and LIPP~\cite{DBLP:journals/pvldb/WuZCCWX21} supports data updates. LIPP
further eliminates the local search via a novel adjustment strategy to
redistribute keys in each node. SOSD~\cite{DBLP:journals/corr/abs-1911-13014,
DBLP:journals/pvldb/MarcusKRSMK0K20} proposes some benchmarks to evaluate
different learned indexes.

Existing learned index techniques cannot be directly applied to \md{} data, as
there is no natural ordering among \md{} data points. Currently, the prevalent
approach is to apply a \emph{linearization} method to convert the problem into a
one-dimensional search problem, on which existing learned indexes can then be
applied. \zorder~\cite{ZM,RSMI} and row-major order~\cite{flood, li2020lisa} are
most widely used. Other linearization methods use learned linear or non-linear
mapping~\cite{DBLP:conf/icde/LiZSWT020}, e.g., clustering followed by the
distance to the cluster center~\cite{ML}. To reduce the challenges of query
skewness and data correlations, Tsunami~\cite{Tsunami} extends Flood via
partitioning data space and modeling conditional CDFs.
Qd-tree~\cite{DBLP:conf/sigmod/YangCWGLMLKA20} utilizes Reinforcement Learning
to optimize the space partitioning. SPRIG~\cite{DBLP:conf/ssd/ZhangRLZ21} uses a
spatial interpolation function to locate the search range in the grid.
~\cite{IOZ} learns a quadtree structure and applies a \zorder variant on each node. ~\cite{10.1145/3035918.3035934} extends \zorder's bit interleaving method to more general cases and proposes the idea of learning SFC. In their work, they employ a heuristic approach to identify an SFC that is both query-aware and skew-tolerant for a given query pattern. Diverging from their approach,  we leverage ML methods to directly optimize query performance under a given workload of queries to find an SFC with the least estimated query time. 

Similar to most space-filling curve-based approaches, an orthogonal aspect for
learned \md indexes is to preprocess the data using simple linear
transformations or sophisticated non-linear transformations, such as the Rank
Space transformation~\cite{RSMI}.

Traditional \md indexes often recursively decompose the space into disjoint or
overlapping partitions. Grid File~\cite{gridfile}, kd-tree~\cite{kdtree} and
Quadtree~\cite{quadtree} are typical examples of the former category, and
\rtree~\cite{rtree} and its variants~\cite{rstartree, sellis1987r+,
kamel1993hilbert} are typical for the latter category. %
Since \rtree~\cite{rtree} variants have been widely adopted in commercial
systems, there are also proposals to integrate learning into \rtree{}.
AI+R~\cite{MDM2022} employs an ML model to predict the set of leaf nodes for a
window query, together with a backup \rtree{}.
\cite{DBLP:journals/corr/abs-2103-04541} aims at learning the key subtree
splitting routine in \rtree{} construction and adapts to the problem instance.
In order to reduce the search range on each leaf node, \cite{ifindex} embeds an
ML model on the selected sort dimension to accelerate the search procedure.

There are other ways of integrating learning into database indexing.
\cite{DBLP:conf/iclr/DongIRW20} uses ML models to learn a balanced space
partition that preserves spatial proximity well.
LIMS~\cite{DBLP:journals/corr/abs-2204-10028} adopts an ML-based data clustering
method to solve similarity search  in metric spaces.

\section{A Framework for Learned SFCs}
\label{sec:fram-learn-sfcs}

In this section, we first summarize window query processing issues for an SFC to
motivate us to design a family of parameterized SFCs that possess salient
properties for query processing.

\subsection{Window Query Processing with SFCs}

A space-filling curve (SFC) is a method of mapping the \md{} data space into the
\oned data space. As we assume the data points have integer coordinate values
within $[0, 2^K -1]$, this naturally leads to a regular partitioning of \md{}
space into $2^{Kd}$ possible points, or \emph{cells} (SFCs can also be applied
on a coarser granularity. 
  E.g., \Flood{}~\cite{flood} can be deemed as using a fixed SFC on
  \emph{grids}). 
an SFC is a bijective function $f$ between these cells and integers within $[0,
2^{Kd} - 1]$. 
We call $x$ in the \md{} space the \emph{original address} and $f(x)$ as its
corresponding \emph{\zaddress{}} with respect to an SFC $f$. 
Intuitively, as $f(x)$ is a one-dimensional integer, it specifies a way to
traverse all the cells exactly once. SFCs are known for their ability to
preserve the \md{} proximity in the linear
order~\cite{DBLP:journals/sigmod/LawderK01}. 
Hence, they are widely used, especially in applications with a need to linearize
\md{} data such as images, tables and spatial data. Some well-known SFCs are
Hilbert curve, Z-order curve, and Gray curve.

To answer a range query $q$ on $D$, assuming that points in $D$ have been mapped
to the corresponding {\zaddress}es, we can 
\begin{inparaenum}[(i)]
\item compute the query's \oned \zaddress{} range $q_z$, 
\item retrieve every point whose \zaddress{} falls within the interval $q_z$,
and 
\item filter out those points that do not fall into the query window $q$. 
\end{inparaenum} 
The tightest \zaddress{} range can be defined as: $[\min_{x \in q} f(x), \max_{x
\in q} f(x) ]$. However, computing these extreme values is difficult in general.
In the worst case, we may need to enumerate all $x$ within $q$, hence with a
cost proportional to the \emph{volume} of the query \emph{and} beating the
purpose of efficient query processing. For certain SFCs with better properties,
such as the Hilbert curve, we still need to enumerate all $x$ on the boundary of
the query window, hence inducing a cost proportional to the \emph{circumference}
of the query.

Nevertheless, we identify a subclass of SFCs such that the above minimization
and maximization can be computed efficiently in $O(d)$ time and hence do not
depend on the size of the query. 

\subsection{Monotonic Space-Filling Curves}

We will first define the criterion that the mapping function of an SFC is
monotonic, and show that the tightest \zaddress range can be  
computed from the lower and upper bounding points of the query window.

\begin{definition}[Monotonic Function in a \Md{} Space]
  Let $a \preceq b$ defined as true if and only if $\forall i, \coord{a}{i} \leq
  \coord{b}{i}$. %
  Then a function $g$ is monotonic, if for all $a$ and $b$, if $a \preceq b$,
  then $g(a) \leq g(b)$. 
\end{definition}

\begin{thm}\label{thm:z-range} If an SFC corresponds to a monotonic mapping
  function $f$, given a spatial query rectangle $q$, the query result $r
  \subseteq \set{x \mid x \in D \land f(q_L) \le f(x) \le f(q_U)}$. In other
  words, the tightest \zaddress{} range of $q$ can be efficiently computed as
  $[f(q_L), f(q_U)]$. 
\end{thm}

  \begin{proof}
    Let $\qmin \definedas \min \set{x \in q}$ with respect to $\preceq$ (i.e.,
    $\qmin$ is $q_L$). Then by definition for any $x \in q, \qmin \preceq x$. As
    $f$ is monotonic, then $f(\qmin) \leq f(x)$. Similarly, we can show that
    $\qmax \definedas \max \set{x \in q}$ (i.e., $\qmax$ is $q_U$) and for any
    $x \in q, f(x) \leq f(\qmax)$. 
  \end{proof}

\begin{example}
  Among three commonly used SFCs, only the \zorder{} curve has the monotonic
  property. We give counter-examples for the Hilbert curve and the Gray curve in
  Figure~\ref{fig:non-monotone}.
  \begin{figure}[htbp]
    \centering
    \includegraphics[width=\linewidth]{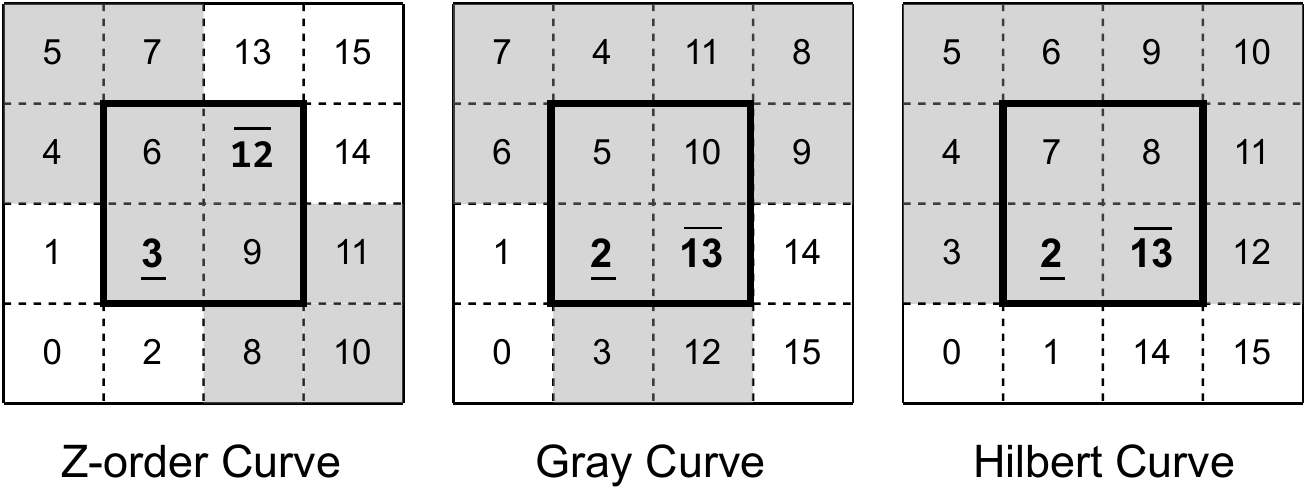}
    \caption{Hilbert and Gray Curves are not Monotonic (Bold Black Rectangle is
      the Query Window; the Tightest \zaddress{} Ranges are Marked in Bold
      Font)}
    \label{fig:non-monotone}
  \end{figure}
\end{example}

\subsection{Parameterized Z-Order SFCs}

Inspired by the monotone property of the \zorder curve, we identify a family of
monotonic SFCs that generalizes the \zorder{} curve. In addition, any instance
within the family has favorable properties to enable efficient query processing. 

We consider the following SFC parameterized by a parameter $\theta =
[\coord{\theta}{1}, \ldots, \coord{\theta}{d}]$: 

\begin{align}
  f(x; \theta) = \sum_{i=1}^{d} \sum_{j=1}^{K} \coord{\theta}{i}_{j} \cdot
  \coord{x}{i}_{j},
  \label{eq:LSFC-framework}
\end{align}
where each $\coord{\theta}{i}$ is a $K$-dimensional vector, $\coord{x}{i}_{j}$
represents the $j$-th bit of the $i$-th dimension value of $x$, $d$ is the dimensionality and $K$ is the maximum number of bits for
$\coord{x}{i}$. In fact, $\coord{\theta}{i}_{j}=2^l$ indicates that
$\coord{x}{i}_{j}$ will be mapped to the $(l+1)$-th bit of the binary
representation of $f(x; \theta)$. If the content is clear, we use $f(x)$ for
short instead of $f(x; \theta)$.

\begin{figure}[tb]
  \centering
  \subfigure[Original Z-order]{
    \includegraphics[width=.29\linewidth]{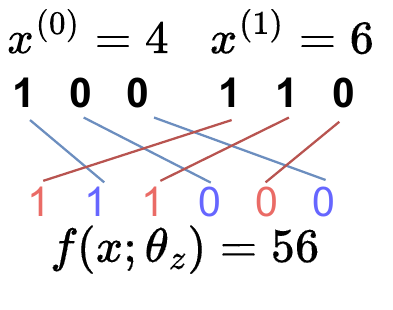}
  
  } \subfigure[Generalized Z-order]{
    \includegraphics[width=.29\linewidth]{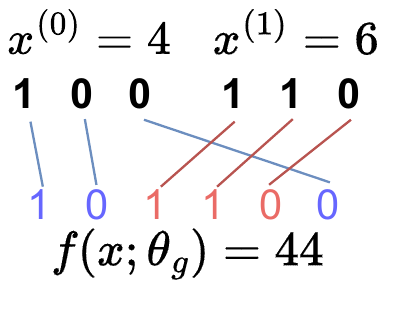}

  } \subfigure[Column-major order]{
    \includegraphics[width=.29\linewidth]{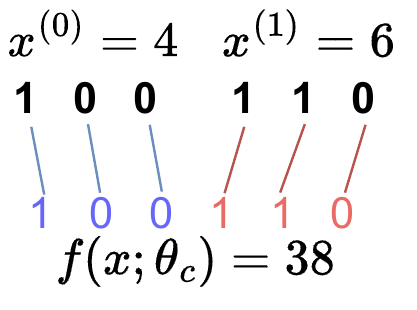}

  } 
 
  \caption{{\zaddress{}} Calculation for Several SFCs within our Parameterized
  SFC Family}

  \label{fig:framework_bit_level}
 
\end{figure}

\begin{example}
  Figure~\ref{fig:framework_bit_level} demonstrates several instances of the
  SFCs within our parameterized family. Figure~\ref{fig:framework_bit_level}(a)
  shows the ordered bit-interleaving way of \zorder to compute $f(x)$ for the
  2-dimensional data point $x = (4, 6)$. Here, $K=3$, and its parameter
  $\theta_z = [\coord{\theta}{1}, \coord{\theta}{2}] = [[1, 4, 16], [2, 8,
  32]]$, and the resulting \zaddress is 56. %
  Figure~\ref{fig:framework_bit_level}(b) demonstrates a new SFC for the same
  data point, but with $\theta_g = [[1, 16, 32], [2, 4, 8]]$, $f(x; \theta_g) =
  44$.
  Finally, Figure~\ref{fig:framework_bit_level}(c) demonstrates yet another SFC,
  which is known as the column-major order, with $\theta_c = [[8, 16, 32],[1, 2,
  4]]$, $f(x; \theta_c) = 38$.
\end{example}

To ensure the family of SFCs preserves the monotonic and bijective properties,
it suffices to impose the following constraints on $\coord{\theta}{i}_{j}$s:
\begin{enumerate}
\item $\coord{\theta}{i}_{j} \in \set{2^0, \ldots, 2^{Kd-1}}$.
\item $\forall {i, j, i', j'}, i \neq i' \lor j \neq j'$, then
  $\coord{\theta}{i}_{j} \neq \coord{\theta}{i'}_{j'}$. 
\item $\forall {j < j'}, \coord{\theta}{i}_{j} < \coord{\theta}{i}_{j'}$. 
\end{enumerate}

  \begin{proof}
    According to the first two constraints, $f(x)$ can be represented as a
    binary form that is actually the combination of the bits from all dimensions
    of a multidimensional data point $x$. Thus, $f(x)$ can be uniquely
    determined by given $x$ and vice versa. So we can conclude $f$ is bijective. 
    Next, we prove $f(x)$ is monotonic. Assuming we have two different
    multidimensional data points $p$ and $q$, the resulting mapped addresses are
    $f(p)$ and $f(q)$ respectively. When $p \preceq q$, according to
    Definition~3, on any dimension $i$ we have $p^{(i)} \leq q^{(i)}$. If
    $p^{(i)} < q^{(i)}$ on dimension $i$, we denote the left-most bit where
    $p^{(i)}$ and $q^{(i)}$ are different by $l$, so we have $p^{(i)}_l=0$ and
    $q^{(i)}_l=1$. The third constraint can ensure the bit order of each
    dimension of a data point $x$ in $f(x)$ remains unchanged. Thus, in the
    binary forms of $f(p)$ and $f(q)$, the left-most different bit is also 0 and
    1, respectively, and we can conclude that $f(p) \le f(q)$, which shows that
    the function $f$ is monotonic.    
  \end{proof}

Both bijective and monotonic properties are \emph{important conscious choices}
in our work to balance (1) the potential of exploiting the spatial locality via
learning an SFC and (2) retaining properties that facilitate efficient query
processing. Violating any constraint breaks the properties of the mapping
function. As an example, consider $\theta = [[1, 4], [2, 10]]$, where the first
constraint is violated. Take $f(x) = 8$. There is no point that can be mapped to
this z-address, indicating that the mapping function is injective but not
bijective. In another example, when $\theta = [[1, 4], [2, 2]]$, the second
constraint is not satisfied. Consider data points $A = (1, 1)$ and $B = (1, 2)$,
then $f(A) = f(B) = 3$. Therefore, the mapping function is not bijective.
Lastly, when $\theta = [[1, 4], [8, 2]]$, the third constraint is violated.
Given $A = (1, 1)$ and $B = (2, 2)$, then $f(A) = 9 > f(B) = 6$, but $A \preceq
B$ according to Definition 3. Therefore, the mapping function is not monotonic.

We note that the Z-order curve is hence a special instance of the above
monotonic SFC. %
Specifically, it corresponds to $\coord{\theta}{i}_{j} = 2^{(j-1) \cdot d +
(i-1)}$. %
We also note that computing the above $f(\cdot)$ is efficient as $f(x)$ can be
computed by ``scrambling'' the bits of $x$ according to $\theta$ using bit
operations efficiently.

\begin{figure*}[htbp]
  \centering
  \subfigure[Row-major Order]{
    \includegraphics[width=.25\linewidth]{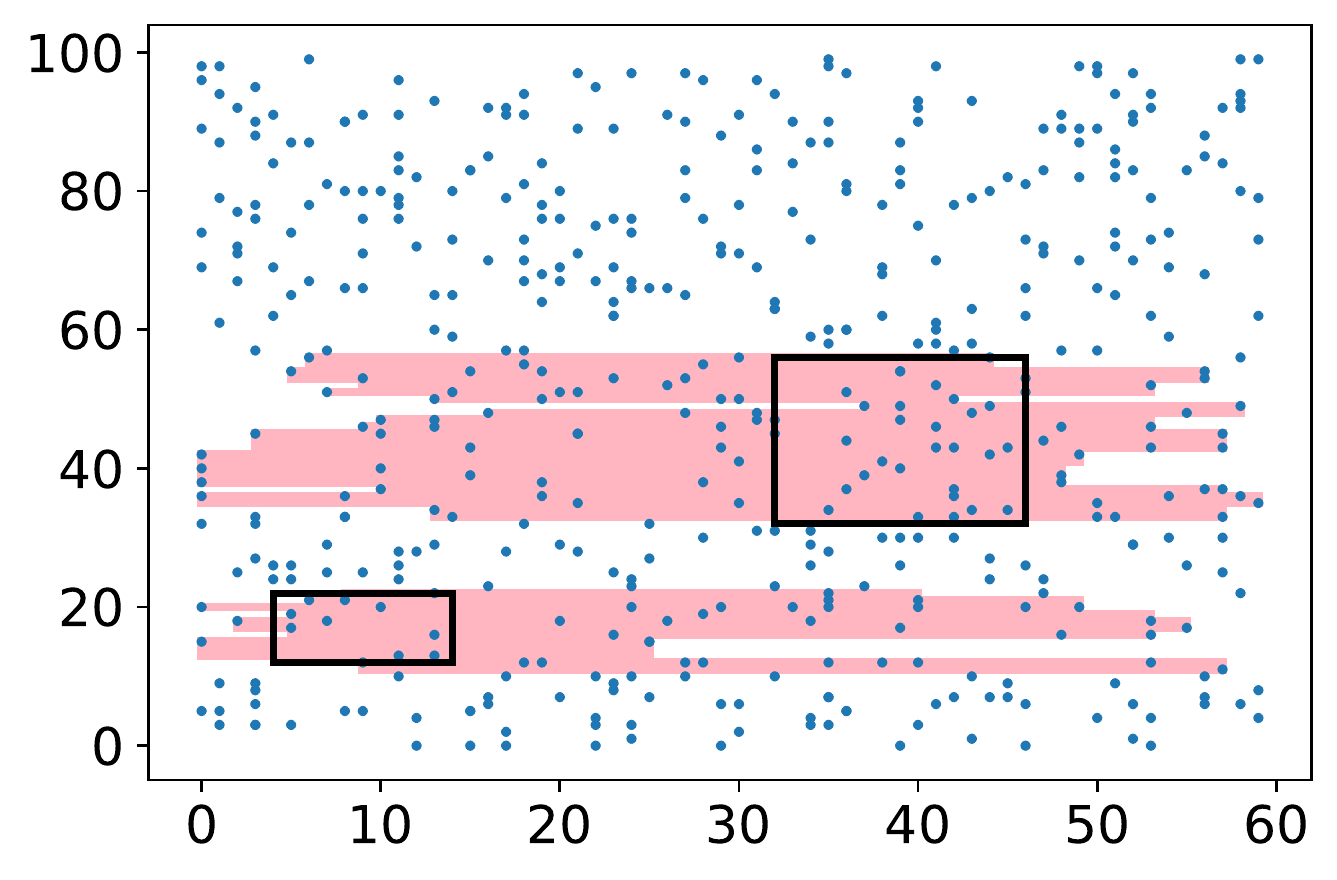}
    \label{fig:cmp-sfc-row}
  } \subfigure[Column-major Order]{
    \includegraphics[width=.25\linewidth]{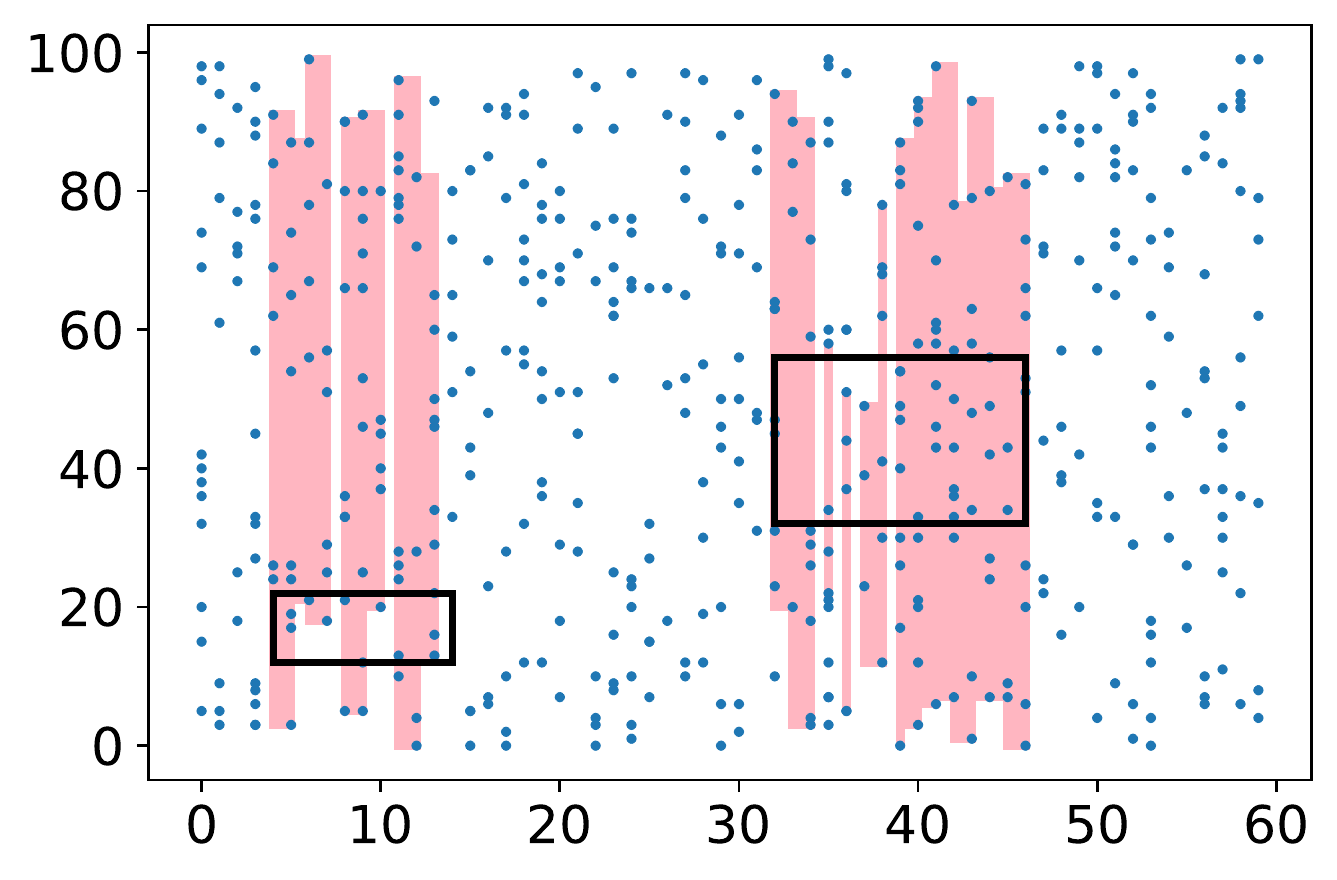}
    \label{fig:cmp-sfc-col}
  } \subfigure[\zorder{} (i.e., \ZM)]{
    \includegraphics[width=.25\linewidth]{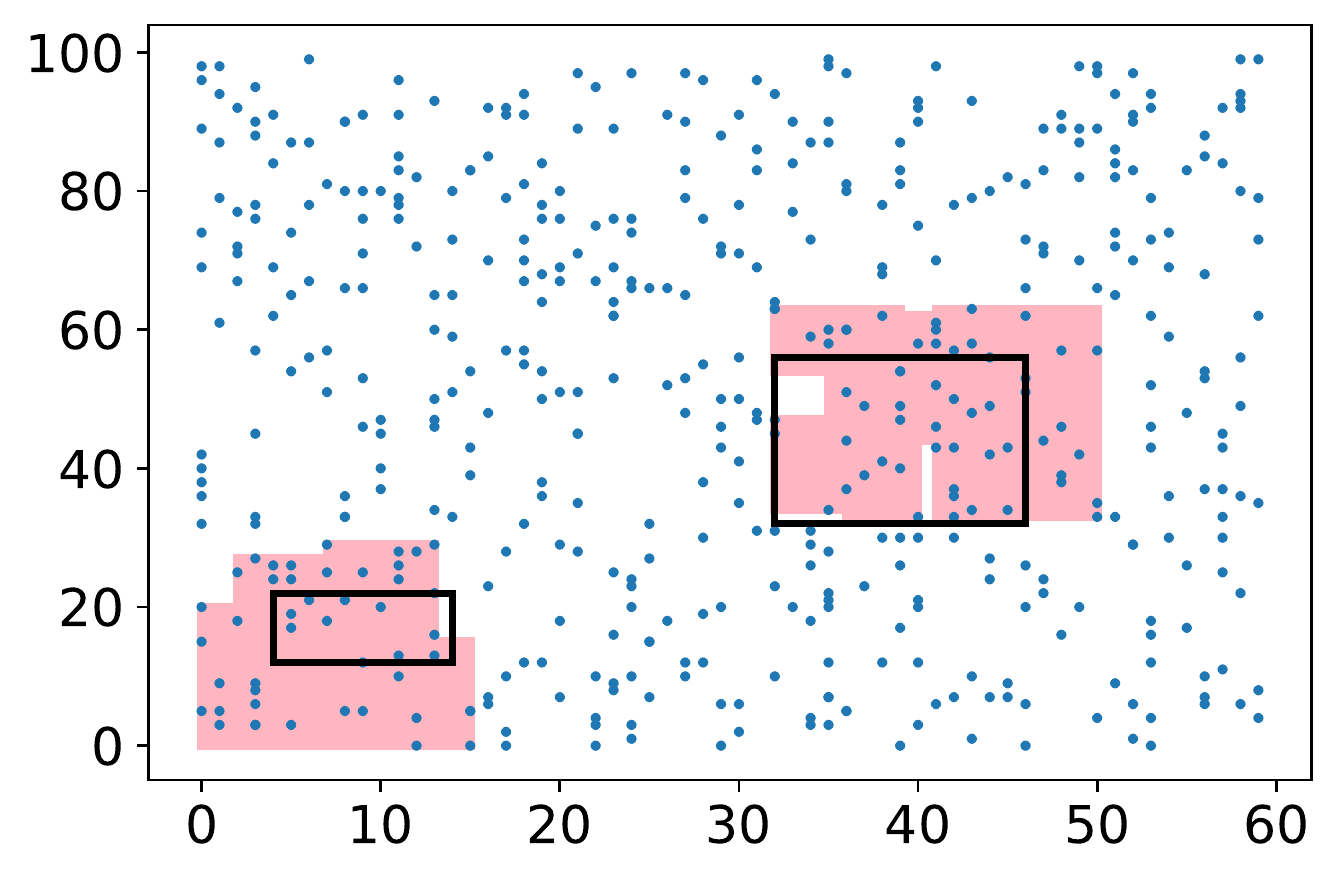}
    \label{fig:cmp-sfc-zm}
  } \subfigure[Learned Grid + Row-major Order (i.e., \Flood)]{
    \includegraphics[width=.25\linewidth]{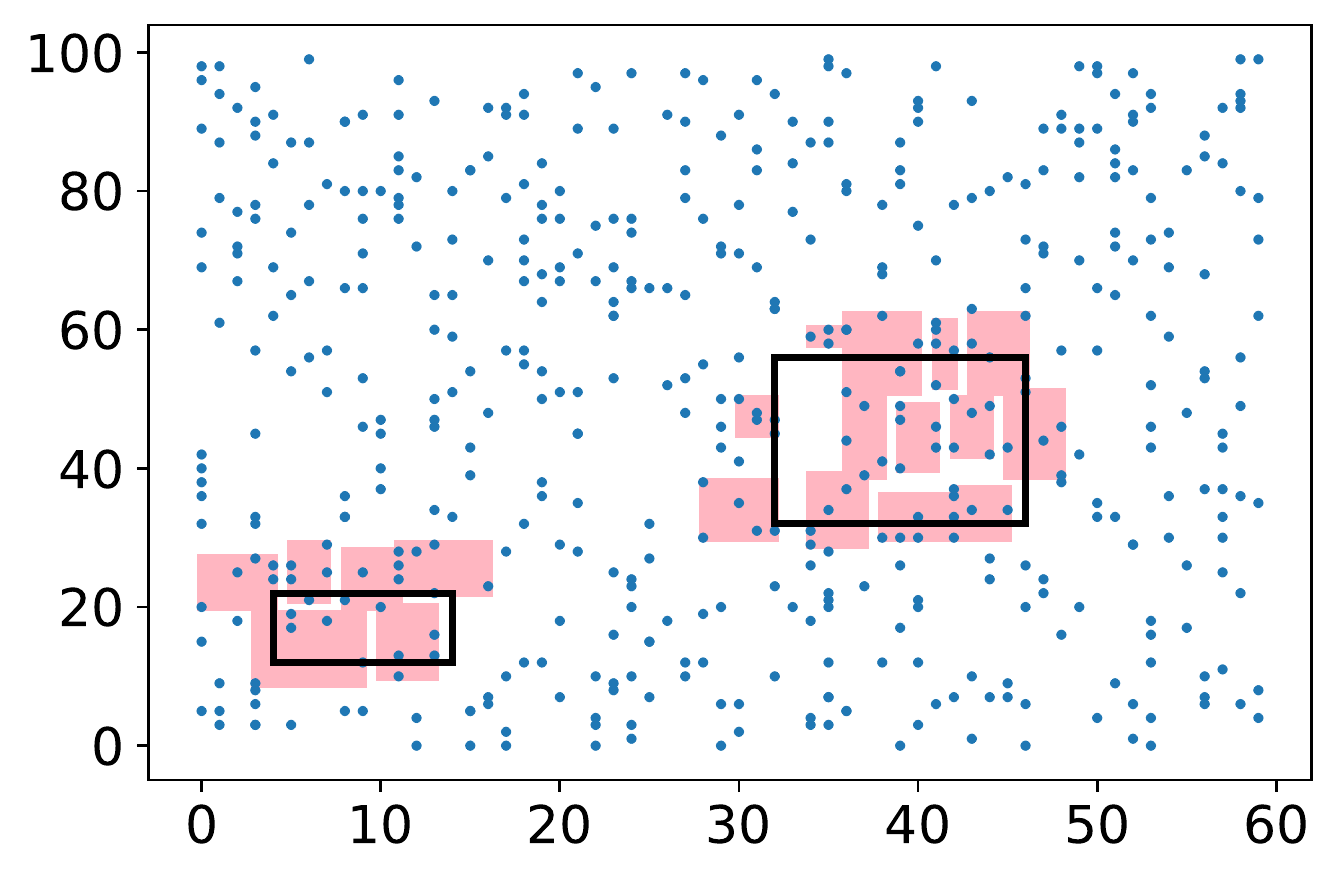}
    \label{fig:cmp-sfc-flood}
  } \subfigure[Learned Monotonic SFC (i.e., \MODELNAME)]{
    \includegraphics[width=.25\linewidth]{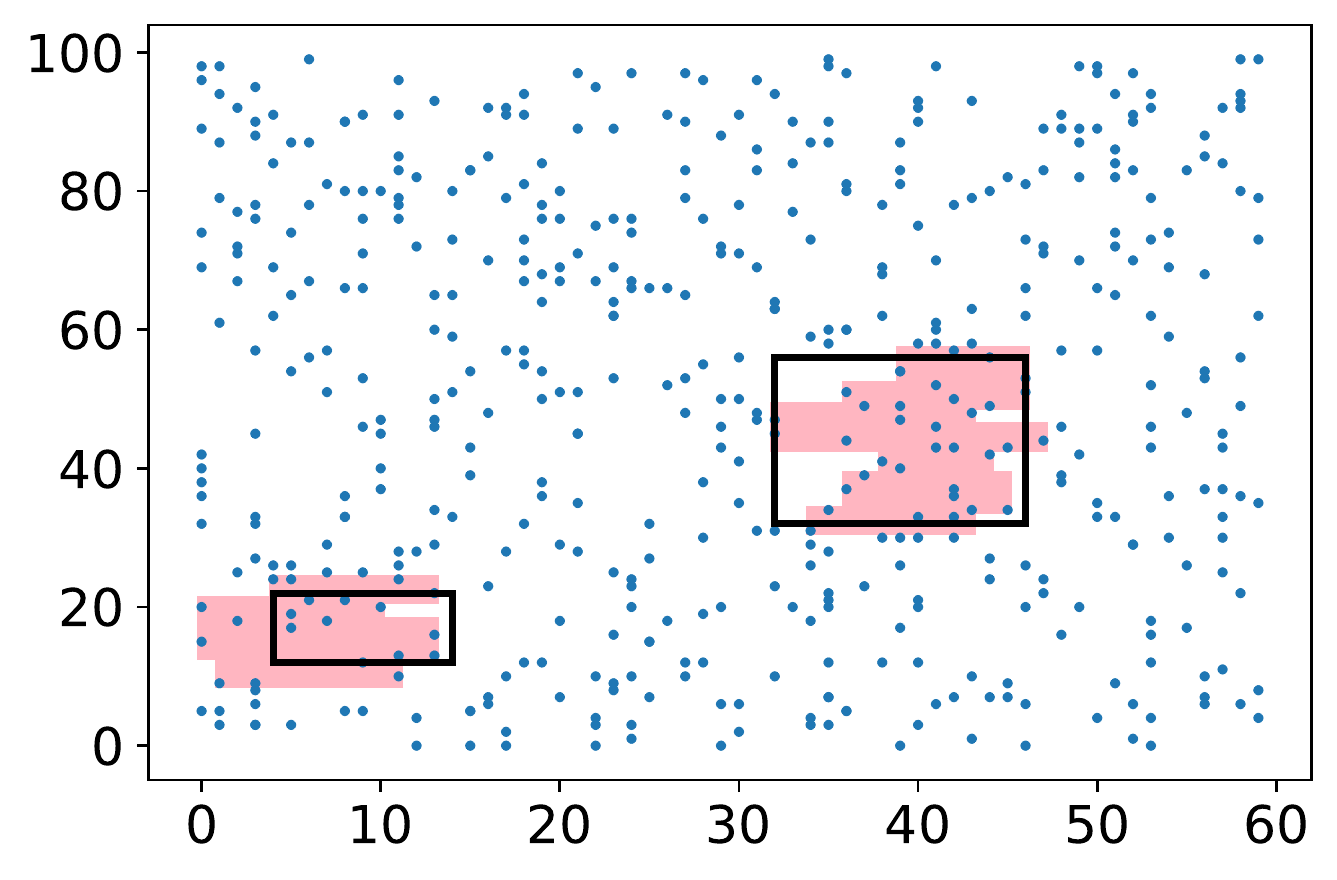}
    \label{fig:cmp-sfc-LMSFC}
  }
  \caption{Visualizing Five SFC-based Indexes (Thick black boxes are queries,
    and data pages accessed during query processing are shaded in red)}
  \label{fig:cmp-sfc-index}
\end{figure*}

As different SFCs induce different linear ordering of the data points in the
dataset, they will result in different query processing costs for a given
workload. This motivates our learning of a good SFC (detailed in the next
Section). Below we provide an example with visualizations to demonstrate this. 

\begin{example}\label{ex:five-sfc} In Figure~\ref{fig:cmp-sfc-index}, we build
  and visualize query processing costs on five indexes based on different
  linearization methods on the same dataset and query workload. We generated a
  set of \emph{random} points in a 2D space, and used the two randomly generated
  queries (denoted as thick black boxes) as the query workload. %
  
  Figures~\ref{fig:cmp-sfc-row} and Figures~\ref{fig:cmp-sfc-col} are the usual
  row-major and column-major order, respectively. Figure~\ref{fig:cmp-sfc-zm}
  uses the \zorder{}, which results in the \ZM{}. Figure~\ref{fig:cmp-sfc-flood}
  is the \Flood{}, which is based on learnable grid partitioning of the space,
  and then orders the grids using a row-major order (with a rearranged dimension
  order). Finally, Figure~\ref{fig:cmp-sfc-LMSFC} shows our proposed method
  \MODELNAME{}, which is based on a learned monotonic SFC.

  We shade the data points accesses during the query processing in
  red.\footnote{For all the methods in the figure, we pack a fixed number of
  data points into each page and the shaded area are based on pages whose MBRs
  overlap with the query window, hence the seemingly irregular shape of the
  shaded areas. %

  } %
  As we can see from Figure~\ref{fig:cmp-sfc-index}, learned indexes, i.e.,
  \Flood{} and our \MODELNAME{}, access fewer data points thanks to the learning
  on the specific instance.
\end{example}

\section{\MODELNAME Index Construction}
\label{sec:index-construction}

\subsection{Overview}

\begin{figure}[htbp]
  \centering
  \includegraphics[width=0.98\linewidth]{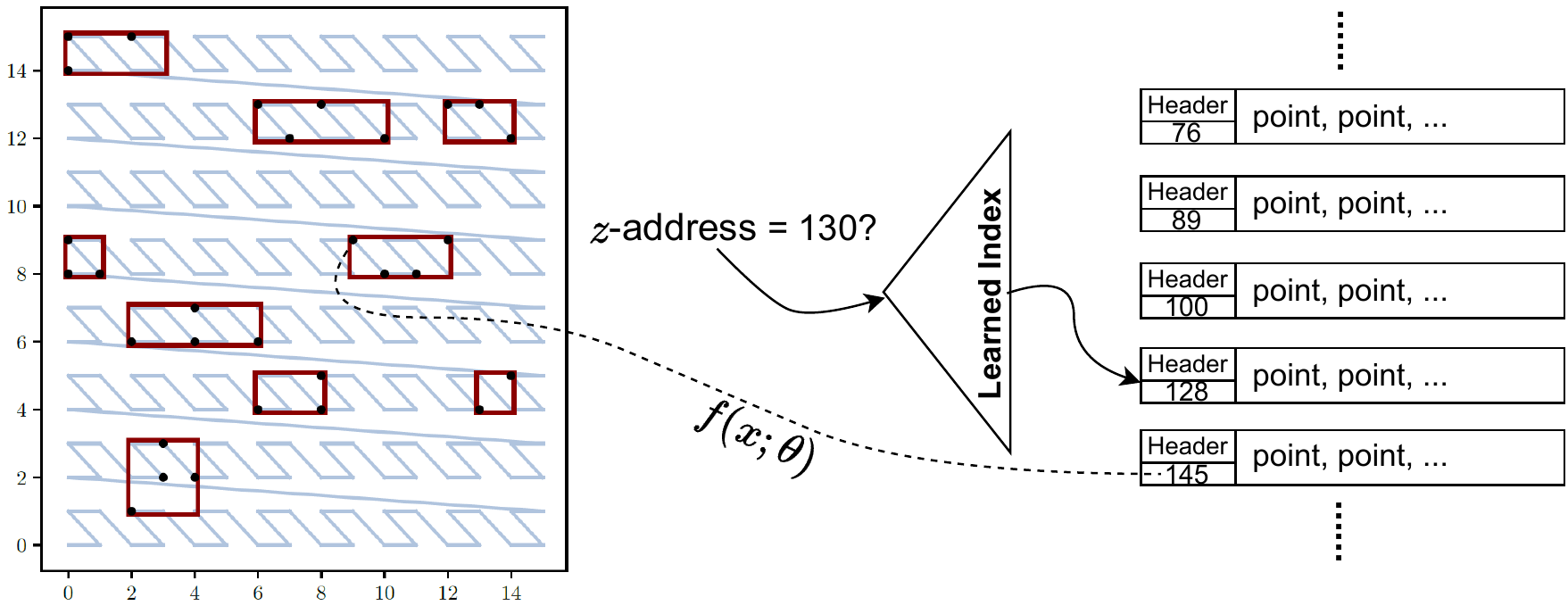}
  \caption{Overview of \MODELNAME{} Index Construction}
  \label{fig:overview}
\end{figure}

Based on the parameterized SFC family defined in the previous section, we
propose a novel \md{} index, \MODELNAME{}, based on \underline{L}earned
\underline{M}onotonic \underline{S}pace \underline{F}illing \underline{C}urves.

We give a high-level sketch of our \MODELNAME{} method in
Figure~\ref{fig:overview}. In our method, we first learn a good monotonic SFC,
$f(x; \theta)$ in Figure~\ref{fig:overview}, which minimizes the query
processing cost of the sampled query workload. The learned SFC gives us a total
order for the \md{} data points. We then propose both an optimal dynamic
programming-based and a sub-optimal yet faster heuristic paging algorithm to
load data points into pages (the thick red rectangles in
Figure~\ref{fig:overview}). Finally, we extract the smallest \zaddress{}es from
each page to form a sorted array, and employ a state-of-the-art learned index
(e.g., pgm~\cite{pgm}) on that, which facilitates the lookup from \zaddress{}es
to pages.

In addition, we also include several optimizations to speed up the query
processing. Specifically, we present a novel query splitting strategy to
minimize the access to spurious pages due to the dimension reduction effect of
the SFC mapping. We also extend the sort dimension optimization~\cite{flood} to
the page-level granularity.

In the rest of this section, we focus on the three steps (i.e., learning an
optimal SFC, Cost-based Paging and Page-level Sort Dimension) in this section
and leave the query processing related techniques to the next Section.

\subsection{Learning an Optimal SFC}
\label{sec:learning}

The goal of learning a parameterized SFC is to find the $\opt{\theta}$ such that
the resulting query processing cost is minimized, i.e., we formulate this as an
optimization problem as:
\begin{align}
  \opt{\theta} = \argmin_{\theta} \E[q \follows \mathcal{Q}]{\mathit{QueryTime} \left(D, q; \theta\right)}
  \label{eq:opt-theta}
\end{align} 
where $\mathcal{Q}$ represents the distribution of the query workload, $q$ is an
\iid{} sample from the distribution, and $QueryTime()$ function returns the
actual query execution time on the given dataset.

To optimize Equation~\eqref{eq:opt-theta}, we can first apply the finite sample
approximation, i.e., by taking sampled queries from $\mathcal{Q}$ and replacing
the expectation with the sample average. Nevertheless, the $QueryTime$ function
is hard to model or approximate accurately as it includes various complex
optimizations. Furthermore, directly evaluating the $QueryTime$ function incurs
exorbitant costs. 

Also, note that $\theta$ is a high-dimensional discrete parameter (with
constraints), the number of choices of $\theta$ (and hence $f(\cdot)$) is
exponential both in the dimensionality $d$ and in $K$, rendering it impossible
to solve the optimization problem via brute force in practice.

\begin{lem}\label{lem:curve-num} The number of different monotonic SFCs for
  $d$-dimensional space is $\Omega{(d!)^K}$.
\end{lem}

\begin{proof}
  First, we consider a permutation $\pi = (x_1, x_2, \ldots, x_d)$ of the $d$
  dimensions, where each dimension has $K$ bits. For dimension $x_j$, its $K$
  bits can choose any of the remaining $Kd - (j-1) \cdot K$ bits. Therefore, the
  total number of choices is: 

  \small{
  \begin{align*}  
    \binom{Kd}{K} &\cdot \binom{Kd - K}{K} \cdot \ldots \cdot \binom{K}{K} 
     = \frac{(Kd)!}{(K!)^d} 
     = \left(\frac{1}{1} \cdot \ldots \cdot \frac{K}{K}\right) \\ 
     & \cdot \left(\frac{K + 1}{1} \cdot \ldots \cdot \frac{2K}{K}\right)
     \cdot \ldots 
     \cdot \left(\frac{Kd - K + 1}{1} \cdot \ldots \cdot \frac{Kd}{K}\right) 
     \geq (d!)^K
  \end{align*}}
 
\end{proof}

Due to the discrete nature, gradient descent style algorithms, used in previous
learned index work~\cite{flood}, cannot be applied to solve this optimization
problem either. Furthermore, the finite-sample approximate also introduces a
small yet avoidable noise to the optimization.

In view of the above challenges, we adopt the state-of-the-art Bayesian
optimization algorithm, \SMBO{}~\cite{SMBO}, to approximately solve the above
optimization problem and return a high-quality SFC that has a low average query
time.
\SMBO builds a surrogate model to approximate the relationship between the
parameters and the actual objective function value. 

The learning process follows the standard \SMBO learning framework. We use the
Random Forest model as the surrogate model instead of the typical Gaussian
Process model, which improves the learning speed and does not rely on the \md{}
Gaussian assumption and the choice of the kernel function. 
One of the advantages of using a surrogate model is that it is cheap to evaluate
while also capturing the approximation noise in the objective function. 
\begin{algorithm}[htbp]
  \small
  \KwIn{training dataset $D$, training query set $Q$, max iteration $\vn{maxIters}$}
  \KwOut{an optimized order $\theta_{best}$}
  \State {Initialize $\est{M}$, $\theta_{best}$ and $y_{best}$}
  \For{$i=0$ to $\vn{maxIters}$} {
    \StateCmt{$\Theta_{cands} = \fn{SelectCands}(\est{M}, \theta_{best})$}
    {selects candidates based on the surrogate model}
    \StateCmt{$Y_{cands} = \fn{BatchEval}(\Theta_{cands}, D, Q)$}
    {evaluates the objective function i.e., $QueryTime$}
    \State{Update $\est{M}$ based on $(\Theta_{cands}, Y_{cands})$}
    \If{ $min(Y_{cands}) < y_{best}$} {
      \State {$y_{best} \gets \min(Y_{cands})$}
      \State {$\theta_{best} \gets \argmin(Y_{cands})$}
    }
  }
  \Return{$\theta_{best}$}
  \caption{SMBO learning framework}
  \label{alg:bo}
\end{algorithm}

Algorithm~\ref{alg:bo} illustrates the \SMBO learning framework, where $\est{M}$ is the surrogate model we need to learn. 
Firstly, the surrogate model is built between the initial candidates and their performance based on the $QueryTime$ evaluation results (line 1).  
Secondly, the \SMBO{} algorithm
uses an acquisition function (e.g., the Expected Improvement~\cite{SMBO}) computed on the
surrogate model to suggest other candidates for evaluation in the next iteration
by automatically balancing exploitation and exploration (line 3). Then, 
we evaluate the objective function by building the index on sampled datasets to
save the evaluation time (line 4). By default, we conservatively use the 5\% sampled dataset as the training dataset, which can maintain
both low query time and learning cost. 
More learning process experiments are
presented in Section~\ref{exp:training process}.
The surrogate model gets updated during each iteration (line 5). Finally, we choose the candidate with the least cost in terms of the objective function.

\subsection{Cost-based Paging}
\label{sec:paging}

In order to accommodate external I/Os and allow for extra
optimizations\footnote{e.g., optimizations based on the MBR and/or the sort
dimension of the pages.}, we need to perform \emph{paging}, which partitions the
dataset $D$ into multiple pages. As usual, we assume that each page has a
maximum size of $B$ bytes and must satisfy a min fill factor constraint
specified by $f \in (0, 1]$. \footnote{Technically, we do allow at most one page
to occupy less than $fB$ bytes.} 
That is, the number of bytes used in each page must be within the range of $[fB,
B]$ bytes.

Finding optimal paging for \md{} dataset is
NP-hard~\cite{DBLP:conf/sigmod/YangCWGLMLKA20}, hence existing \md{} indexing
methods usually perform paging based on some heuristics, e.g., \rtree and its
variants used the heuristics to minimize the margin, dead space and overlap area
of the MBRs of the resulting pages~\cite{rstartree, sellis1987r+}. Not
surprisingly, this practice is inherited by \md{} indexes based on SFCs. For
example, RSMI~\cite{RSMI} simply loads the maximum number of points into each
page, which we term as \emph{fixed-sized paging}.

We observe that paging is important and actually can be solved \emph{optimally} 
for SFC-based \md{} indexes. 
This is because we can record the MBRs of the data points within each page and
use the MBR to further optimize the query processing. On one hand, a page can be
skipped if a page's MBR is disjoint with a query; on the other hand, if a page's
MBR is contained in a query, we can process the data points on the page 
sequentially without other filtering overhead. In both cases, such optimizations
are more likely if the MBR of a page is small. Default one-dimensional paging
methods, such as the fixed-size paging method, are not aware of the MBR of the
pages and cannot perform active optimizations for it.

Based on the above observations, we design a scoring function $S(P)$ that is
intuitively the density of a page $P$, or $S(P) = \frac{vol(P)}{size(P)}$, where
$vol(P)$ and $size(P)$ gives the volume of MBR of the page $P$ and the number of
data points in the page $P$, respectively. 

We then formulate the \emph{optimal cost-based paging} problem as finding a
paging solution, i.e., a partitioning $P \definedas \set{P_1, \ldots, P_{k(P)}}$
over $D$ (where $k(P)$ denotes the number of resulting pages), such that the
total score of the solution $P$ is minimized, i.e., 

{\small
\begin{align*}
  \opt{P} = \argmin_{P}  \sum_{{j \in \set{1, \ldots, k(P)}}} S(P_j), \, \text{subject to } size(P_j) \in [fB, B]
\end{align*}
}

In the following, we first give an algorithm to solve the above problem
optimally based on Dynamic Programming (DP), and then give a sub-optimal but
fast heuristic paging algorithm. Both methods achieve a better paging layout
than the fixed-sized paging and hence improve the query performance.

\subsubsection{Dynamic Programming Paging Method}
\label{sec:dynam-progr-paging}

Thanks to the SFC which provides a linear order for the data points, we are able
to circumvent the NP-hardness of the multi-dimensional paging problem by solving
the one-dimensional paging problem optimally via dynamic programming. %

Let $OPT[i]$ be the optimal cost obtained by an optimal cost-based paging
algorithm for the first $i$ data points. Then we can derive the following
recurrent equations:

{\small   
\begin{align*}
    &OPT[i] = S(Page(D[1 \twoldots i])) && , i < \frac{fB}{4d}\\
    &OPT[i] = \min_{\mathclap{s \in [\frac{fB}{4d}, \frac{B}{4d}]}} \left( OPT[i - s] + S(Page(D[i-s+1\twoldots i]) \right) && , \text{otherwise}
\end{align*}
}

\noindent, where $Page(z)$ denotes the page formed by a set of points, denoted as $z$ and we
assume each integer takes 4 bytes. %
Obviously, $OPT[n]$ gives the cost of the optimal paging for the entire dataset,
and it is easy to use backtracking to report the optimal paging solution
$\opt{P}$. The time complexity of the dynamic programming paging method is
$O(\frac{nB}{4d})$ as the scoring function $S$ can be computed in $O(1)$ time via
incremental computation. The above algorithm is implemented in Algorithm~\ref{alg:dp_paging_method}.

\begin{algorithm}[htbp]
  \small 
  \KwIn{$D$ is a $d$-dimensional data array
  ordered by $Z$-address, a fill factor constraint $f$, page size limit $B$, 
  and score function $S$} 
  \KwOut{Return a page array $P$} 
  \State{$i \gets 1$}
  \State{initialize $OPT$ array} 
  \While{$i \le |D|$}{
    \If{$i < \frac{fB}{4d}$} {
      \State{$\displaystyle OPT[i] \gets S(Page(D[1 \twoldots i]))$} 
    } 
    \Else {
      \State{$\displaystyle OPT[i] \gets \min_{{s \in [\frac{fB}{4d}, \frac{B}{4d}]}} \left( OPT[i - s] + S(Page(D[i-s+1 \twoldots i])) \right)$}
    }
    \State{$i \gets i + 1$}
  }
  \State{$P \gets$ using backtracking on $OPT$} 
  \Return{$P$} 
  \caption{$DynamicProgrammingPagingMethod$} 
  \label{alg:dp_paging_method} 
\end{algorithm}

\subsubsection{Heuristic Paging Method}

Although the DP algorithm is linear in $n$, it is still time-consuming in
practice as $B$ is typically a large constant (e.g., $B = 8192$ in our
experiment). %
There, we further propose a heuristic paging method, which can achieve
comparable query performance and faster construction time compared with the DP
method. %

The heuristic algorithm is a greedy packing algorithm, which packs as many data
points into the current page as possible until some condition is violated. 
We give the pseudo-code of heuristic paging method
 in Algorithm~\ref{alg:heuristic_paging_method}. 
The condition stipulates that the new MBR, formed by adding the current data point
into the page (line 6), should \emph{not} enlarge the old MBR by more than $\alpha$ times
($\alpha > 1$ is a hyper-parameter) (line 7). This condition reduces the chance that the MBR of the resulting page becomes too large (with respect to the number of data points within), where a large MBR may cause much dead space and increase the chance of intersecting with the queries. 

 \begin{algorithm}[htbp]
  \small
  \KwIn{$D$ is a $d$-dimensional data array
  ordered by $Z$-address, a fill factor constraint $f$, page size limit $B$, and 
     $\alpha$ is a float, used to control page's MBR}
  \KwOut{$P$} 
  \State{Initialize $P$ array}        
  \State{Initialize a empty page $p$}
  \While{Unvisited Points in $D$}{
    \If{$p$ is empty} {
      \State{Load $\frac{fB}{4d}$ data points into $p$ and update $p$'s mbr}
    }
      \State{$mbr' \gets$ the MBR of $p$ containing the next data point}
      \If{$size(p) < \frac{B}{4d}$ \textbf{\textup{and}} $Volume(mbr') < \alpha \cdot
      Volume(p.mbr)$} { 
        \State{add the next point in $p$ and update $p.mbr$} 
      }\Else {
        \State{$P.add(p)$}         
        \State{Initialize a empty page $p$}
      }
  
    } 
    \If{$\size{p} > 0$}
    { \State{$P.add(p)$} } 
    \Return{$P$} 
    \caption{$HeuristicPagingMethod$}
  \label{alg:heuristic_paging_method}
\end{algorithm}

\subsection{Page-level Sort Dimension}

Following \Flood~\cite{flood}, we maintain the points in each page sorted in a
chosen dimension named \emph{sort dimension}. Unlike \Flood where the sort
dimension is fixed for \emph{all} pages, we allow using different sort
dimensions in different pages, which provides more skipping opportunities to
filter as many irrelevant points as possible when processing intersecting pages.
This is because different sort dimensions may result in various sizes of the
search area after refinement. Thus, we can choose the sort dimension that can
achieve the least search cost for each page. In a similar vein, we utilize the
query workload information to choose the sort dimension for each page as
follows: for each page, we collect the set of intersecting queries in the query workload.
We estimate the query cost using each of the $d$ dimensions as the sorting
dimension, and choose the one with the least query cost. If there are no
intersecting queries for a page, we use a default order, which is determined in
the same way as \Flood{}.

Once we have determined the sort dimension for each page, we order the data
points in the page by increasing order of the sort dimension. When we need to
scan a page, we can first refine the search range (or physical storage range)
according to the corresponding sort dimension. Specifically, given a window
query $q$, the range constraint over sort dimension $d^{*}$ is
$\coord{q_L}{d^{*}} \leq \coord{x}{d^{*}} \leq \coord{q_U}{d^{*}}$. The points
in each page are stored contiguously in increasing order of the corresponding
sort dimension. Thus, we can use binary search or a one-dimensional index model
to 
  accelerate the search via finding the lower-bound position of
  $\coord{q_L}{d^{*}}$ and the upper-bound position of $\coord{q_U}{d^{*}}$ in
  the physical storage range. As a result, points that do not satisfy the range
  constraint on the sort dimension are filtered out, which reduces the scanning
  overhead.  

Besides, sort dimension can also reduce the computation cost during
verification. Once we determine the search range, we can guarantee the points in
this range are satisfied the query's constraint over the sort dimension.
Therefore, there is no need to verify the value on the sort dimension, resulting
in saved computation overhead.

\section{Query Processing}
\label{sec:query-processing}

As alluded to in Section~\ref{sec:fram-learn-sfcs}, to answer a query $q$ using
an SFC-based index, it takes two steps:

\begin{enumerate}[(1)]
\item Projection: given a spatial query rectangle $q$, we determine its scan
  range in the \zaddress{}es as $[f(q_L), f(q_U)]$ according to
  Theorem~\ref{thm:z-range}. 

\item Scan with filtering: All pages whose \zaddress{}es fall within the range
  need to be scanned. We can locate these pages using the forward index.
  Additionally, as we maintain the MBR and page-specific sort dimension information for each
  page, we can perform folklore optimizations such as skipping irrelevant pages,
  and scanning only the relevant portion of the page.

\end{enumerate}

The main limitation of the above framework is that it ignores the potentially
large numbers of false positive data points within the \zaddress{} range due to
the SFC mapping.

\begin{figure}[tb]
  \centering
  \subfigure[Learned Z-order w/o query splitting]{
    \includegraphics[width=.45\linewidth]{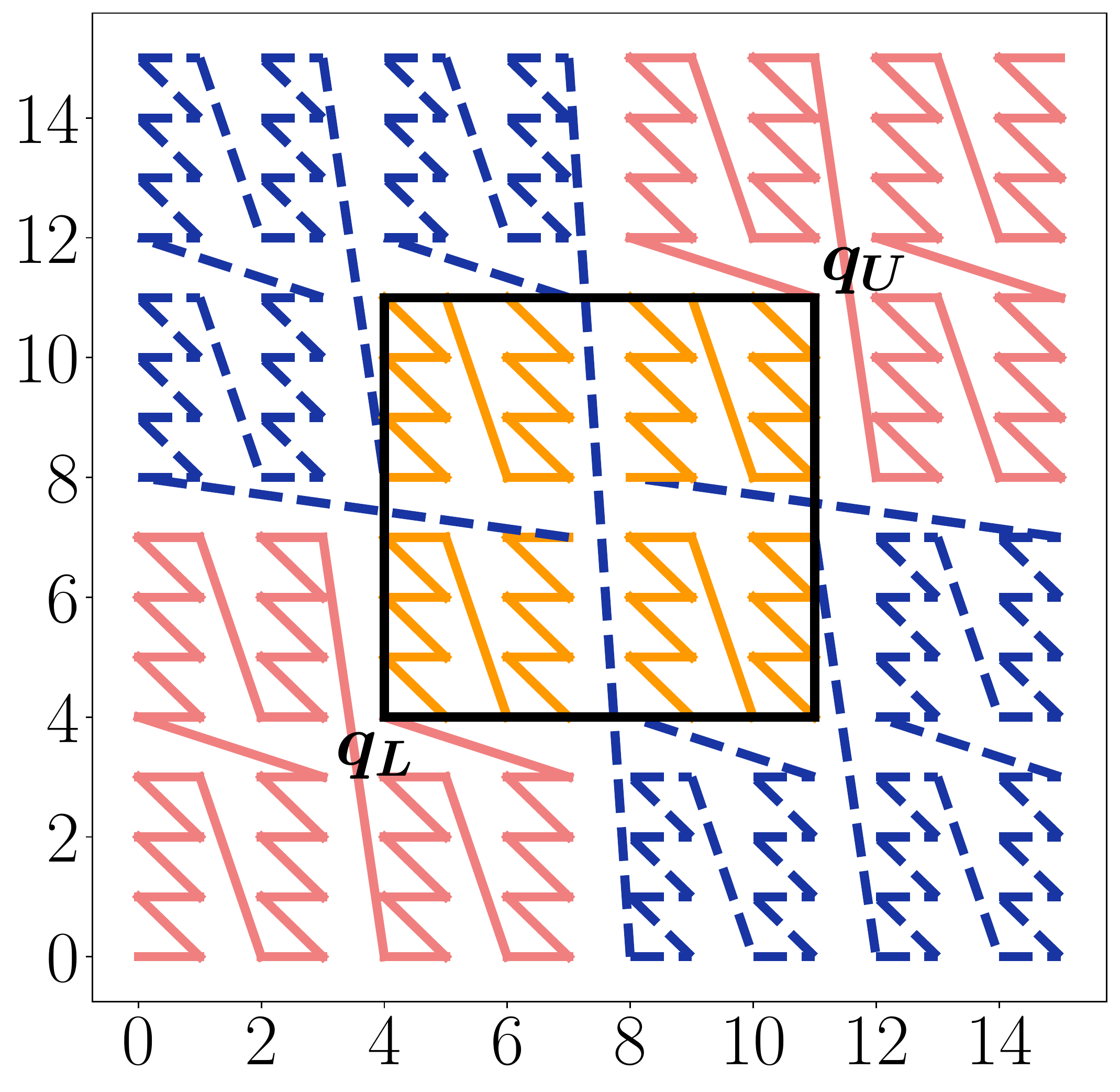}
  }
  \subfigure[Learned Z-order with query splitting]{
    \includegraphics[width=.45\linewidth]{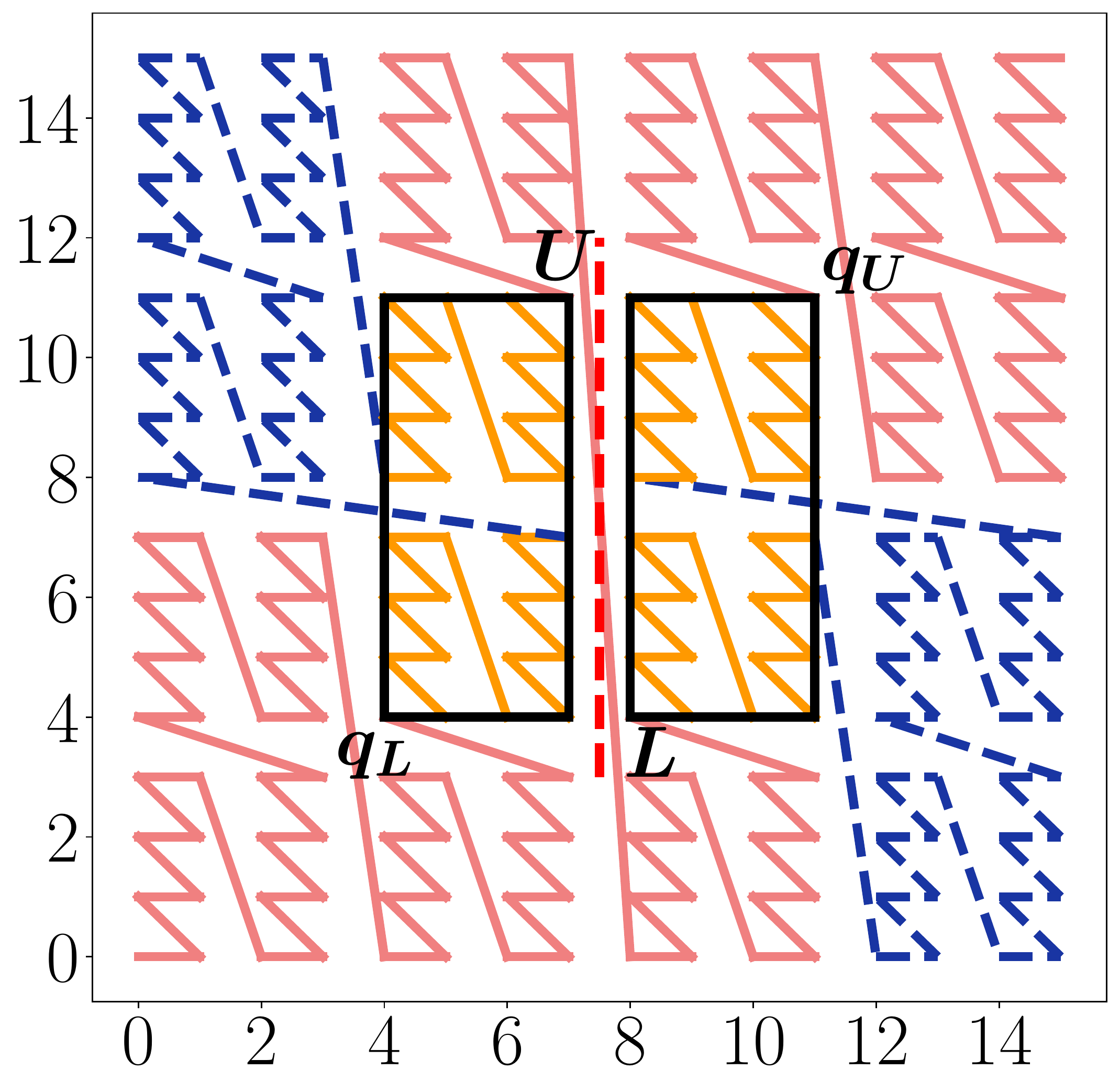}
  }
  \caption{Example of query partition (black rectangle is a query window, yellow 
    part only contains relevant points within the search range, blue part
    only contains irrelevant points within the search range)}
  \label{fig:query_partition}
\end{figure}

\begin{example}
  Consider the query (the black rectangle) in
  Figure~\ref{fig:query_partition}(a). The corresponding \zaddress{} range can
  be decomposed into the blue part\textbf{s} (false positives) and the
  yellow part\textbf{s}. Scanning the whole range, even with filtering, incur much
  unnecessary overhead in accessing and filtering pages that contain only false
  positive points. 
\end{example}

This phenomenon was also observed in a few methods such as
UB-Tree~\cite{DBLP:conf/vldb/RamsakMFZEB00}, and a lazy skipping strategy was
used. In this strategy, instead of scanning all the pages in the \zaddress{}
range, it invokes a skipping function, $\fn{FindNextZaddress}$, after scanning
the current page, to compute the next page that contains the \emph{first} true
positive point with respect to $q$. While this strategy is guaranteed to skip
\emph{all} the false positive \emph{pages}, it has the following drawbacks: 
\begin{inparaenum}[(I)]
\item It incurs significant overhead as this function has to be invoked for
  every true positive page. $\fn{FindNextZaddress}$, technically, will compute
  the next \zaddress{} (which may be virtual and does not correspond to any
  data point) after that of the last point in the current page, and translate it
  to a page. 
  This will require accessing the forward index. As we employ a learned index,
  this incurs a non-trivial overhead for model estimation and local search. Even if
  we employ a \bplustree{}, it may require accessing internal or leaf pages of the
  index. 
\item The skipping function may only skip very few pages; in fact, in many
  cases, it will just return the next page.

\end{inparaenum}

Instead, we propose a novel proactive skipping strategy based on query
splitting, which is especially efficient on monotonic SFCs such as ours. We
illustrate its idea in the following example.

\begin{example}
  Consider the same query (the black rectangle) in
  Figure~\ref{fig:query_partition}(a). We can split the query into two parts by
  cutting it at the value $8$ on the $x$-axis, as illustrated in
  Figure~\ref{fig:query_partition}(b). We still plot the false positive parts in 
  blue. It reduces the number of false positive parts compared with the
  case without query splitting --- e.g., the part $[4, 7] \times [12, 15]$ is
  eliminated.
\end{example}

\subsection{Recursive Query Splitting}
\label{section:hqpm}

We start by introducing a procedure to compute the best way to split a query
into exactly two sub-queries (i.e., optimal 1-split), and then we generalize it
to obtain multiple sub-queries based on recursive splitting.

\paragraph{Optimal 1-Split Algorithm}

Consider a query $q$ that corresponds to a \zaddress{} range $[f(q_L), f(q_U)]$.
Without loss of generality, assume that we split at the value $v$ on the $\delta$-th
dimension. This will split the query into two sub-queries, with the
corresponding \zaddress{} ranges as $[f(q_L), f(U)]$ and $[f(L), f(q_U)]$,
respectively (See Figure~\ref{fig:query_partition}(b)). We define the \emph{cost} of the split as $f(U) - f(q_L) + f(q_U) - f(L)$, or
intuitively, the sum of the \zaddress{} ranges of the two resulting sub-queries.
This cost function is chosen as it is highly correlated with the actual query
processing cost after the split and can be easily computed \emph{without}
accessing the data points. %

Then we formulate the \emph{optimal 1-split problem} as finding the split (i.e.,
the dimension and the value) such that the cost of such split is the minimum.

Or formally, 
\begin{align*}
  (\opt{\delta}, \opt{v})= \argmin_{\delta \in [1, d], v \in [\coord{q_L}{\delta}, \coord{q_U}{\delta}]} f(U) - f(q_L) + f(q_U) - f(L)
\end{align*}
Note that both $U$ and $L$ are determined by $\delta$ and $v$, but we omit the
notational dependency for the easy of exposition. %

As $f(q_L)$ and $f(q_U)$ are constants for the fixed query $q$, the above
minimization is equivalent to the following maximization problem, i.e., finding
the maximum ``gap'' between $f(U)$ and $f(L)$:
\begin{align*}
  \argmax_{\delta \in [1, d], v \in [\coord{q_L}{\delta}, \coord{q_U}{\delta}]} f(L) - f(U)
\end{align*}
Plugging in the definition of $f(\cdot)$, it becomes: 
\begin{align}
  \argmax_{\delta \in [1, d], v \in [\coord{q_L}{\delta}, \coord{q_U}{\delta}]} %
  \sum_{i = 1}^{d}\sum_{j = 1}^{K} \theta_j^{(i)} \cdot (L_j^{(i)} - U_j^{(i)})
\end{align}

For a fixed $\delta \in [1, d]$, we can find the optimal split value $\opt{v}$ as:
\begin{align}
  \opt{v} & = \argmax_{v \in [\coord{q_L}{\delta}, \coord{q_U}{\delta}]} \sum_{j = 1}^{K}
            \theta_j^{(\delta)} \cdot (L_j^{(\delta)} - U_j^{(\delta)})
            + C \notag\\ 
          & = \argmax_{v \in [\coord{q_L}{\delta}, \coord{q_U}{\delta}]} \sum_{j = 1}^{K}
            \theta_j^{(\delta)} \cdot (L_j^{(\delta)} - U_j^{(\delta)})    \label{eq:one-d-cut}
\end{align}
where $C = \left(\sum_{i \neq \delta}^{d}\sum_{j = 1}^{K} \theta_j^{(i)}
\cdot (L_j^{(i)} - U_j^{(i)}) \right)$ is a constant.

\begin{lem}\label{lem:optimal-one-cut}
  A solution to the optimization problem of Equation~\ref{eq:one-d-cut} is
  $(q_U^{(\delta)} \mathtt{>>} l) \mathtt{<<} l$, where $l$ is the most
  significant bit of $\coord{q_L}{\delta} \oplus \coord{q_U}{\delta}$, where
  $\oplus$ denotes $\mathrm{XOR}$. 
\end{lem}

Therefore, we can solve the optimal 1-split problem by finding the optimal cut
value for each of the $d$ dimensions, hence, the complexity is only $O(d)$. 

We note that Lemma~\ref{lem:optimal-one-cut} holds because of the fact that
$\coord{U}{\delta} + 1 = v = \coord{L}{\delta}$ and the fact that
$\theta_{j+1} \geq 2 \theta_{j}$ (derived easily from the constraints introduced
to guarantee the monotonic property). If the monotonic property does
\textbf{not} hold, then one may need to check every possible $v$ value to
perform the optimization, hence taking $O(\| q_U - q_L \|_{1})$ time complexity,
which means the resulting procedure may be more expensive for ``large'' queries.

\begin{example}
  Consider the example in Figure~\ref{fig:query_partition}(a) again. The query 
  $q = [4, 11] \times [4, 11]$, and the learned SFC corresponds to the parameter
  \begin{align*}
    \theta = [[2^0, 2^3, 2^5, 2^7], [2^1, 2^2, 2^4, 2^6]]
  \end{align*}
  Hence, $q$'s \zaddress{} range is $[f(q_L), f(q_U)] = [48, 207]$. 
  Our optimal 1-split algorithm will first consider the $x$-axis. In this case,
  the most significant bit $l = 3$ as $\mathtt{(0100)_2 \oplus (1011)_2 = (1111)_2}$. Then
  the $\opt{v}$ on the axis is $\mathtt{(1011)_2 >> 3 << 3 = (1000)_2} = 8$, and then
  the cost of splitting at 8 on the $x$-axis can be calculated. 

\end{example}

\paragraph{Recursive Splitting}

  As one split is often insufficient to reduce the number of irrelevant pages,
  we adopt our optimal 1-split algorithm recursively to divide the query window
  into multiple parts. In our implementation, the stopping condition is set as
  either reaching a recursion depth of $k_{\textrm{maxsplit}}$ or when there is
  no gap to split. $k_{\textrm{maxsplit}}$ is a parameter that can
  balance the number of index accesses with the skipping opportunity of disjoint
  pages. A higher $k_{\textrm{maxsplit}}$ can effectively filter out disjoint
  pages but causes more index access overhead. Conversely, a lower
  $k_{\textrm{maxsplit}}$ saves the cost on index access but may not eliminate
  enough irrelevant pages. We provide the pseudo-code of Recursive Query Splitting 
  in Algorithm~\ref{alg:RQS}. In each recursion, we find the optimal cut for each $d$ dimension and then we choose the split with the maximum ``gap'' (line 5).  Based on the split value $\opt{v}$ and corresponding dimension $\delta$, a spatial rectangle can be split into two and passed into the next iteration (lines 6-8).

\begin{algorithm}[htbp]
  \small
  \KwIn{$q$ is a spatial rectangle, $k$ is an integer, used to represent the number of remaining partition times}
  \KwOut{Return an array of search range $qs$}
  \State{Initialize $qs$ array}
  \If{$k = 0\,\lor$ there is no gap to split $q$} {
    \State{$qs.add(q)$}
  }
  \Else{
    \State{Find the optimal cut $\opt{v}$ and corresponding dimension $\delta$ based on $Lemma\,2$}
    \State{$q1,\,q2 \gets$ split $q$ based on $\opt{v}$ and $\delta$}
    \State{$r1 \gets Recursive Query Splitting(q1, k - 1)$}
    \State{$r2 \gets Recursive Query Splitting(q2, k - 1)$}
    \State{For each spatial rectangle from $r1$ and $r2$, we add it into $qs$}
  }
  \Return{$qs$}
  \caption{$Recursive Query Splitting$}
  \label{alg:RQS}
\end{algorithm}
\section{Experiments}
\label{sec:exp}

\subsection{Experimental Settings}

\myparagraph{Datasets}%
We use three real-world datasets with different characteristics in our
experiments (See Table~\ref{tab:dataset}) and they are also used in the previous
work. We preprocess the datasets to scale up all coordinates to integers and
remove duplicates.
\textbf{OSM} is a spatial dataset consisting of 250M records randomly sampled
from North America in the OpenStreetMap
dataset\footnote{\url{https://download.geofabrik.de/}}. %
  We use the GPS coordinates (i.e., longitude and latitude) to form a 2D
  dataset.
\textbf{NYC} is randomly sampled from records of yellow taxi trips in New York
  City in 2018 and
  2019\footnote{\url{https://www1.nyc.gov/site/tlc/about/tlc-trip-record-data.page}}.
  We used the pick-up locations, trip distances, and total amounts to form a 3D
  dataset.
\textbf{STOCK} consists of daily historical stock prices from 1970 to
  2018\footnote{\url{https://www.kaggle.com/ehallmar/daily-historical-stock-prices-1970-2018}}.
  We select four features: the high price, the low price, the adjusted close
  price, and trading volume.

\begin{table}[htpb]%
  \small
  \caption{Dataset Characteristics}
  \label{tab:dataset}
 
  \begin{center}
  \begin{tabular}{l|r|r|r}
    \toprule
          & $n$ (\#-of-Points) & $d$ (\#-of-Dimensions) & Size (GB) \\\midrule
    OSM   & 250M              & 2                     & 1.95      \\
    NYC   & 30M               & 3                     & 0.35      \\
    STOCK & 30M               & 4                     & 0.47      \\\bottomrule
  \end{tabular}
\end{center}
\end{table}

\myparagraph{Query Workload}%
As the datasets do not come with their query workloads, we generate the default
query workloads as follows.

A query is parameterized by its center and its range in every dimension. We
generate query centers in one of the two modes:
\begin{inparaenum}[(i)]
\item \emph{Skewed}, in which query centers are randomly sampled data points.
\item \emph{Uniform}, in which query centers are randomly sampled within the
  data space. 
\end{inparaenum}
The width of the queries in each dimension is uniformly sampled from zero to the
width of the data space of that dimension scaled by 0.05. 
All query windows are clipped to be within the data space.

Following~\cite{learnedcard}, the final query workload is obtained by mixing
90\% of the skewed queries with 10\% uniform queries. The resulting
selectivities for the three datasets are about 0.7\%, 0.07\%, 0.01\%,
respectively. For each dataset, we generate a training and a test query workload
of sizes 1000 independently. All queries use the \texttt{COUNT} aggregate
function, i.e., reporting the number of data points within the query window.

\myparagraph{Algorithms}%
We compare our proposed \MODELNAME with the following algorithms:
\begin{enumerate}[$\bullet$]
\item \textbf{\ZM}~\cite{ZM} combines the fixed \zorder curve and learned index,
  together with the fixed-size paging.

\item \textbf{\rstartree}~\cite{rstartree} is a traditional and widely-used
  \md{} index. We use the implement \rstartree{} in the Boost C++ Libraries
  (\url{https://www.boost.org}).

\item \textbf{\Flood}~\cite{flood} is another state-of-the-art learned index for
  \md{} data, which learns an optimal configuration of \md{} grids for a given
  data and query workload. %
  It can be approximately viewed as a variable-size SFC that follows row-major
  order for an appropriate permutation of the $d$ dimensions. %
  We adapt the \Flood{} method with fixed-size paging for a fair comparison.

\end{enumerate}

We note that these algorithms represent prior state-of-the-art indexes in
different categories. For example, \Flood{} has outperformed other \md{}
indexes, such as Grid File~\cite{gridfile}, kd-tree~\cite{kdtree},
UB-tree~\cite{DBLP:conf/vldb/RamsakMFZEB00} and Hyperoctree~\cite{hyperoctree},
in their experimental evaluation. %
We do not consider other learned \md{} indexes such as \LISA{}\cite{li2020lisa}
, \RSMI{}\cite{RSMI}, as the original codes have various
limitations or did not achieve competitive performance in our experiments.

\cite{IOZ} is another related work, which learns a quadtree and applies a
\zorder variant on each node to adapt query workload. One difference is that
\MODELNAME learns the mapping between data points and addresses, so we can
directly locate the search range rather than traversing a tree structure. In
addition, we use a BO algorithm to learn better ordering. Another difference is
that we directly use the actual query time as the metric rather than the number
of false positive points used in \cite{IOZ}. Since the performance of \cite{IOZ}
is even worse than the baseline model (i.e., \zm), we do not include it in our
experiment.

We experimented with \TSUNAMI{}\cite{Tsunami} but did not report its performance
here, because its splitting algorithm does not result in any split in any
dimension on our query workloads\footnote{The skewness of our
workload is based on data distribution while Tsunami is not. }, in
which case, \TSUNAMI{}'s performance degrades to that of \Flood{}.

We use \texttt{C++} to implement all the methods. To compare these methods
fairly, we run all the experiments either in the in-memory mode (for those that
do not support external I/O) or in the warm buffer mode (for those that support
external I/O). %
The page size is set to $B = 8192$ bytes, and the min fill factor $f = 0.25$.
For the one-dimensional space mapped from all SFCs (\zorder{} or learned SFCs),
we use 64 bits, and $K = \floor{\frac{64}{d}}$. 
 And we empirically choose $k_{\textrm{maxsplit}}$ = 4 in
recursive query splitting.

For learned indexes, we use the PGM~\cite{pgm} as a one-dimensional learned
index since PGM achieves the competitive range query performance in the static
dataset and is easily embedded in the different learned \md 
indexes. The error bound in PGM is empirically set to 128, which is robust for
different configurations and datasets.

All the experiments are performed on a machine with i9-7900X CPU @ 3.30GHz  
and 64 GB main memory running Ubuntu 20.04.4.

\subsection{Query Performance}

\begin{figure}[htpb] 
\centering 
\includegraphics[width=0.85\linewidth]{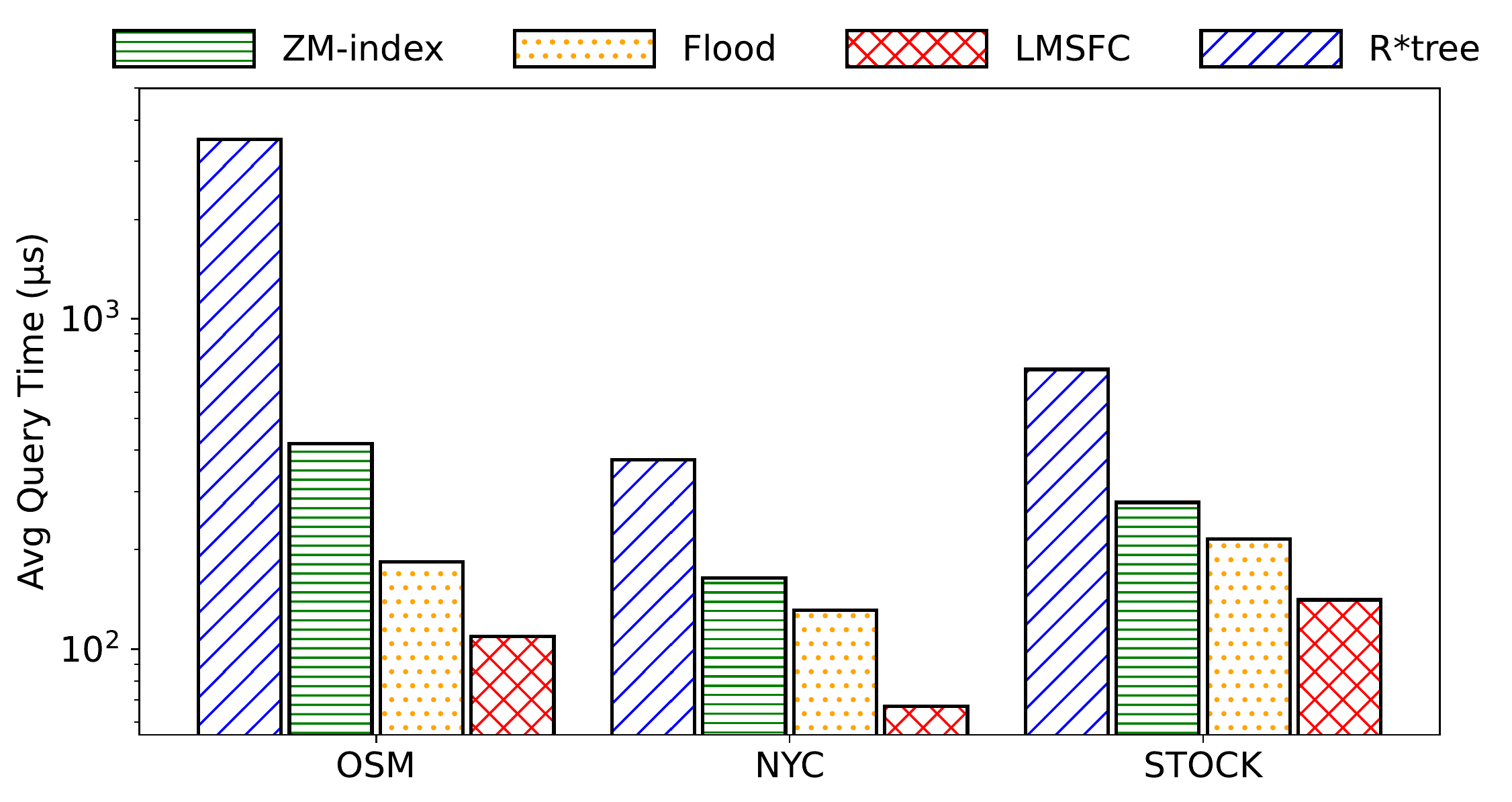} 

  \caption{Query Performance} 
  \label{fig:queryoerformance} 
\end{figure}

 In this section, we compare \MODELNAME with a traditional \md index and other
learned \md indexes for \md range queries. 

Figure~\ref{fig:queryoerformance} reports the query time for different indexes
on each dataset. \MODELNAME outperforms all other indexes across all the
datasets. \MODELNAME achieves between 1.52$\times$ and 1.96$\times$ speedup on
query time compared with the runner-up. 
  The main advantage of LMSFC over the baseline (i.e., the Z-order curve) is
  that the learned SFC preserves multi-dimensional locality better after the
  mapping into the one-dimension address space. Consequently, close-by points in
  the multi-dimensional space are more likely assigned into the same page. This
  leads to pages with smaller/compact MBRs, hence fewer page accesses when
  answering range queries. 
We note that \MODELNAME is much faster than the \zm, achieving 3.8$\times$
speedup on the OSM dataset. This demonstrates the huge potential of a learned
SFC versus a fixed SFC as this is the key difference between the two indexes. 
All learned \md indexes are significantly superior to the traditional \md index
\rstartree{}. This confirms that there is a need to incorporate ML-based methods
into database components to improve the performance.

  In addition, we further investigate false positive (FP) records scanned by
  each index. In OSM dataset, the number of scanned FP points by \rstartree,
  \zm, \Flood and \MODELNAME are 60940, 72291, 25947, and 19067 separately per
  query. \zm scans more FP points than \rstartree since \rstartree utilizes a
  heuristic method to achieve good clustering during packing.  By using a query
  workload as prior knowledge, \Flood and \MODELNAME can adapt their structures
  to access fewer FP points. Besides, \MODELNAME applies page optimizations to
  further reduce FP points. Thus, \MODELNAME achieves the smallest number of
  false positive points being scanned. In the other two higher dimensional
  datasets, \MODELNAME shows superiority in this metric, which are 17.4$\times$
  and 11.1$\times$ less than \rstartree, 3.8$\times$ and 5.0$\times$ less than
  \Flood, and 10.1$\times$ and 6.4$\times$ less than \zm. This is because data
  points are sparse in the higher dimensional data space, resulting in poor
  clustering during paging. Thus, paging optimization needs to be considered in
  \md indexes.

\subsection{Selectivity}

\label{sec:selectivity}

\begin{figure}[htpb]
  \begin{minipage}[t]{0.46\linewidth}
    \centering
    \includegraphics[width=\linewidth]{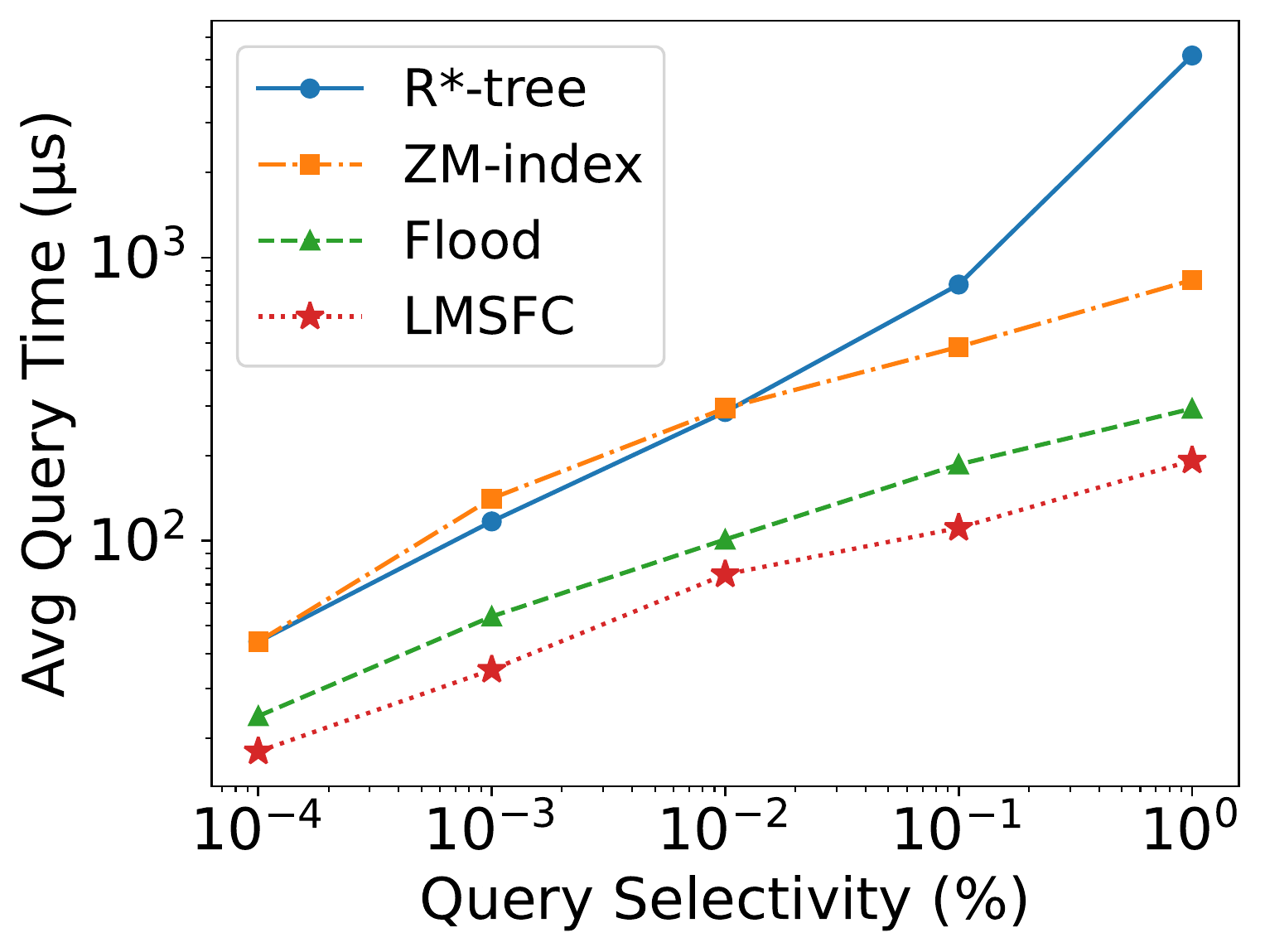}
    \caption{Varying Query Selectivity}
    \label{fig:selectivity}
  \end{minipage}
  \hspace{0.05\linewidth}
  \begin{minipage}[t]{0.46\linewidth}
    \centering
    \includegraphics[width=\linewidth]{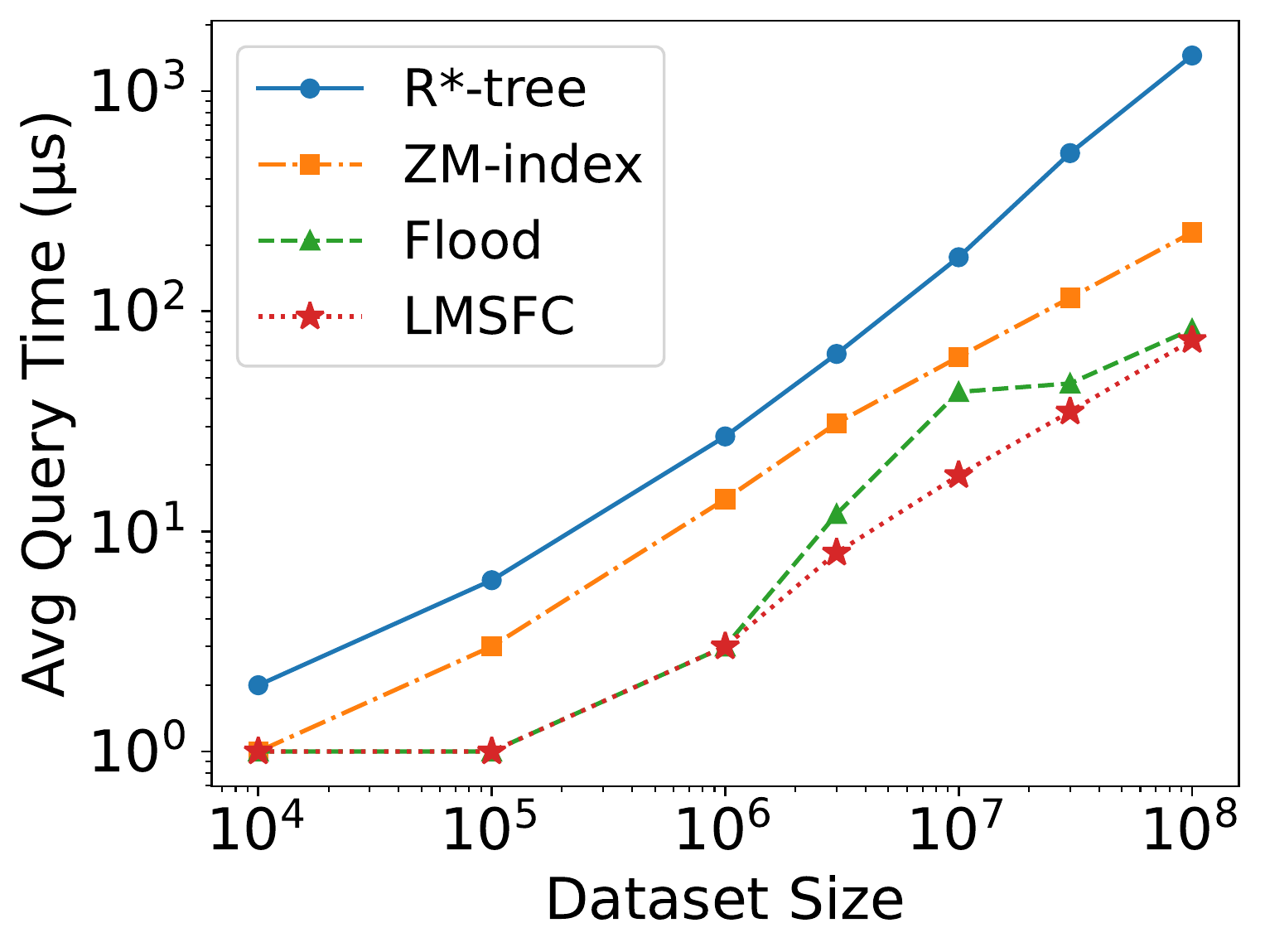}
    \caption{Varying Dataset Size}
    \label{fig:datainc}
  \end{minipage}%
 
\end{figure}

Next, we study the performances of all the methods with respect to query
selectivity. We vary the selectivity from 0.0001\% to 1\% by uniformly scaling
the query windows accordingly. 

Figure~\ref{fig:selectivity} shows the result on the OSM dataset (and similar
results are present in other datasets), where both axes are in logarithmic
scales. %
For all the methods, the query times grow approximately linearly with the query
selectivity since more data points are accessed. We notice a deterioration with
the \rstartree when the query selectivity becomes large; this may be due to the
fact that when the query window grows, it is more likely to intersect with more
pages, and hence more page access and backtracking. %

  Conceptually, learned indexes only need to scan pages within a certain
  \zaddress{} range, and there is no complex intersecting MBR test or
  back-tracking. However, \rtree and its variants need to perform more complex
  MBR intersection queries per inner node and may result in many back-tracking
  (esp., for higher dimensional cases).

\ZM{} achieves similar performance with \rstartree{}, but performs consistently
across the selectivity range. %
\Flood{} has a significant improvement consistently over \ZM{}, while \MODELNAME
further improves the query performance consistently, demonstrating the wide
applicability of learned indexes.

\subsection{Dataset Scalability}
\label{sec:datascale}
To investigate the scalability, we sub-sample the OSM dataset to create datasets
of the same distribution but with varying sizes.  
Figure~\ref{fig:datainc} shows the result, where both axes are in logarithmic
scales. %
We can see that all the indexes scale approximately linearly with the data size,
whereas \MODELNAME{} performs the best, followed by \Flood{}, \ZM{}, and finally
\rstartree{}. %
We also investigated the reason why \Flood{} behaves noticeably worse than
expected for 10M data points. It is partly due to the sub-optimal configuration
it learned from the sampled dataset. If we allow \Flood{} to learn from, e.g.,
30\% of the dataset, the resulting performance matches the approximate linear
trend much more closely.

\begin{figure*}[htpb]
    \centering
    \subfigure[OSM]{
      \includegraphics[width=.3\linewidth]{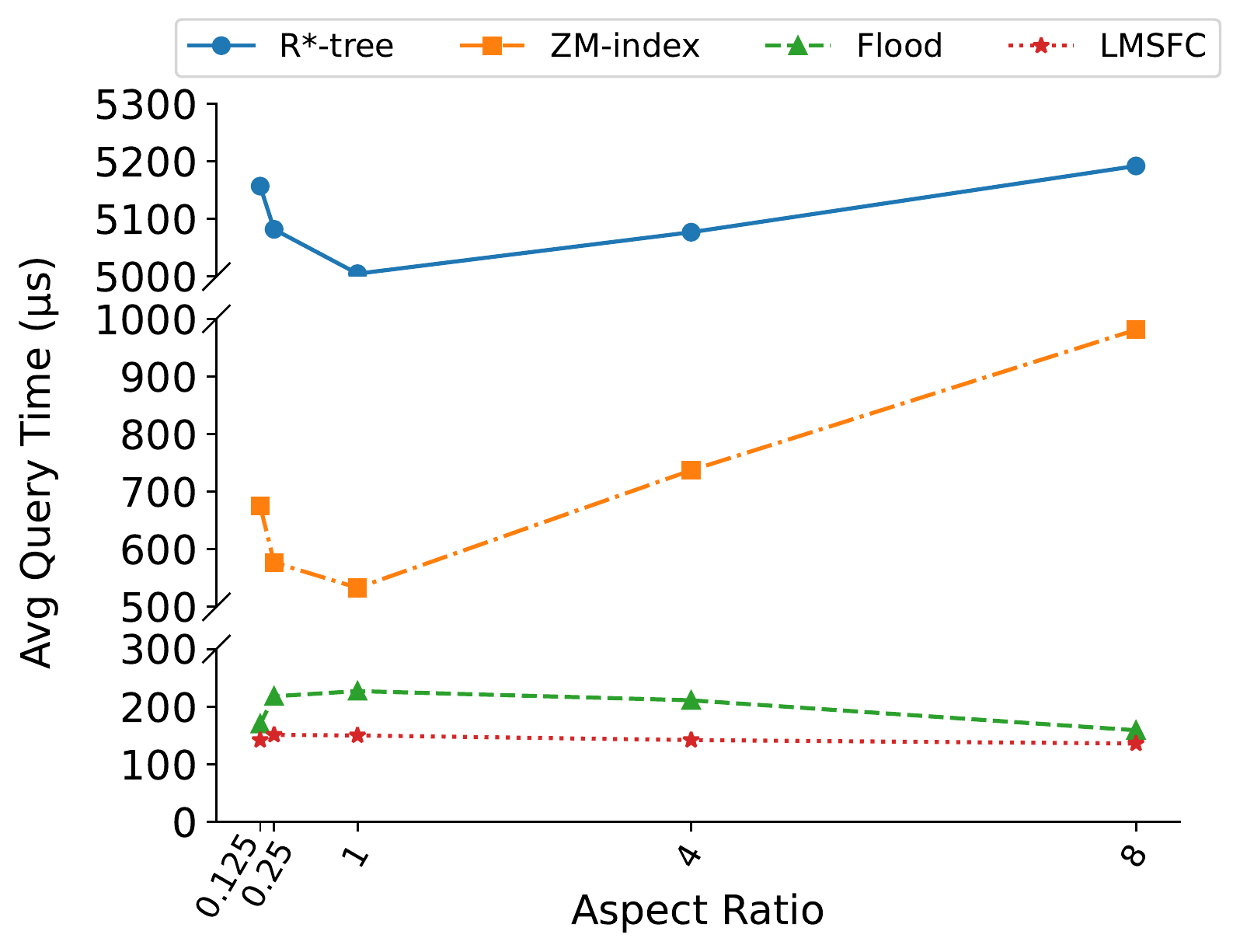}

    } \subfigure[NYC]{
      \includegraphics[width=.3\linewidth]{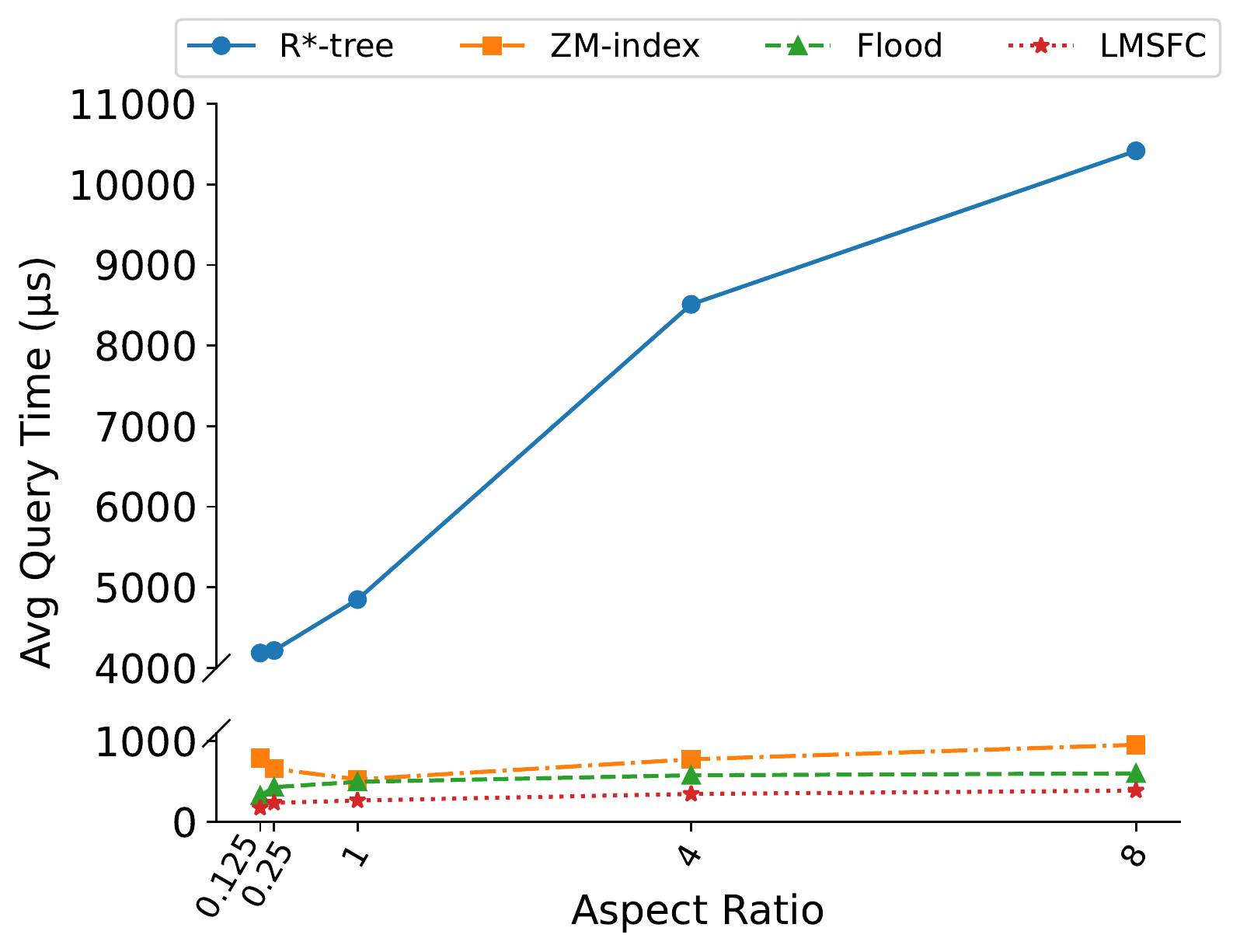}
    } \subfigure[STOCK]{
      \includegraphics[width=.3\linewidth]{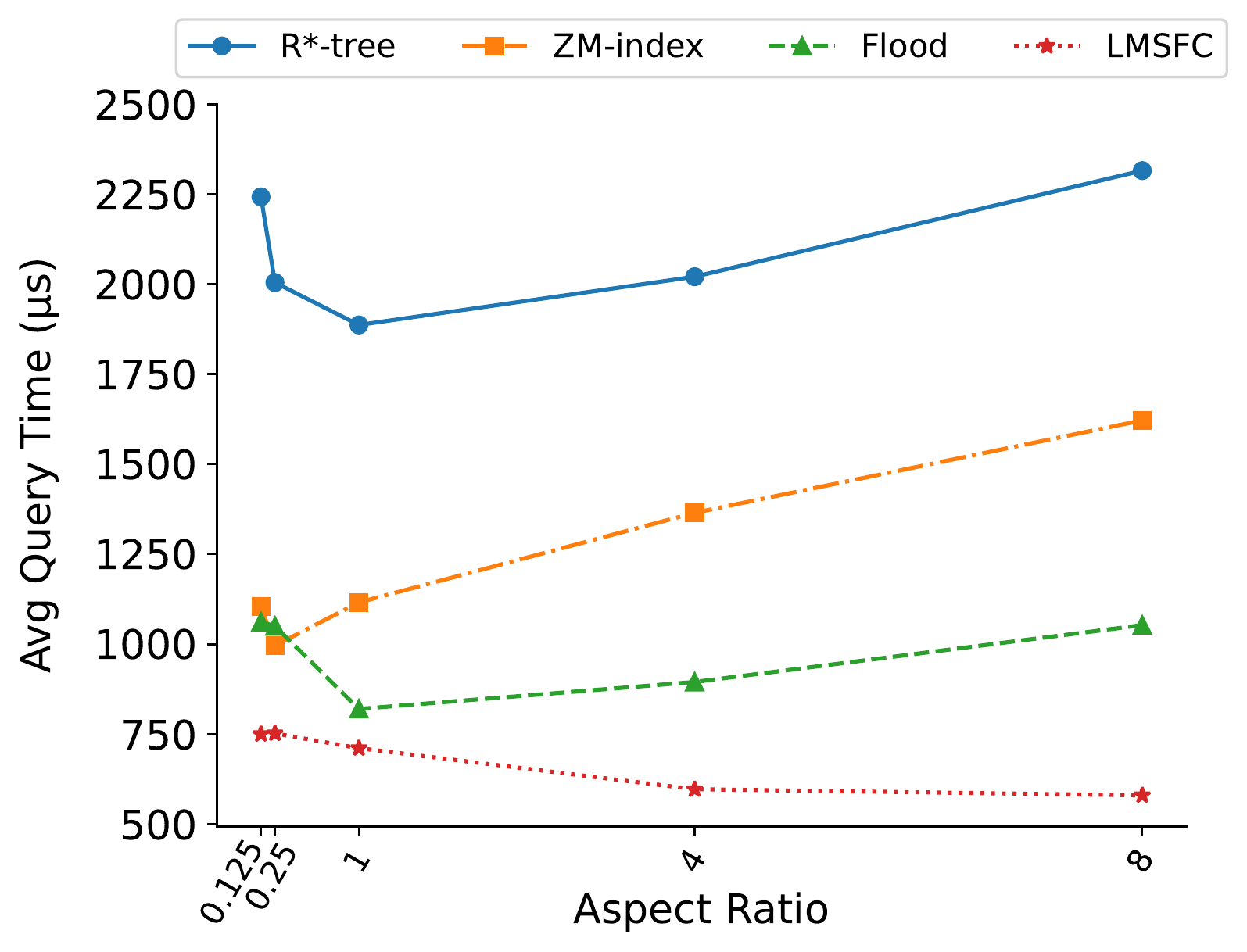}
    }
   
    \caption{Different Aspect Ratio}
    \label{fig:aspectratio}
    
\end{figure*}
\subsection{Aspect Ratio}

We investigate how the performance varies with the aspect ratio of the query
window. We fix the selectivity to be 1\% and then vary the aspect ratio from
0.125 to 8.0. %
The aspect ratio is defined as the ratio of the width of two dimensions of the
query window. For datasets of more than 2D (i.e., NYC and STOCK), we randomly
select one dimension, called \emph{variable dimension}, for a given dataset to
enforce the aspect ratio. We then start with a query window of equal size on all
dimensions, and then modify the length of the variable dimension to satisfy the
ratio constraint. Finally, with the aspect ratio fixed, we scale the query
window to enforce the same selectivity.
For example, the three sides of a 3D query window with ratio of 4 and 0.25 will
have a side length ratio of 4:1:1 and 0.25:1:1, respectively, assuming the first
dimension is the variable dimension.

As shown in Figure~\ref{fig:aspectratio}, \MODELNAME offers the fastest query
speed among all the indexes. The two learned indexes, \MODELNAME{} and \Flood{}
show much more stable performance than the other two non-learned indexes,
demonstrating that learned indexes can adapt well to the query workload to
achieve consistent and superior performance. Furthermore, we notice that
\MODELNAME{} outperforms \Flood{} across all settings, especially in the STOCK
dataset, which is partly due to the fact that \Flood{} has to learn a
$(d-1)$-dimensional grid, which is harder for larger $d$. Finally, we notice
that there are cases where non-learned indexes behave significantly worse even
for ``symmetric'' aspect ratios. E.g., on the NYC dataset, \rstartree{}'s
performance is almost 3x at aspect ratio of 8.0 as compared with that at aspect
ratio of $\frac{1}{8.0}$, demonstrating the need for learned indexes for \md{}
datasets.

\begin{figure}[ht]
  \centering
  \includegraphics[width=0.85\linewidth]{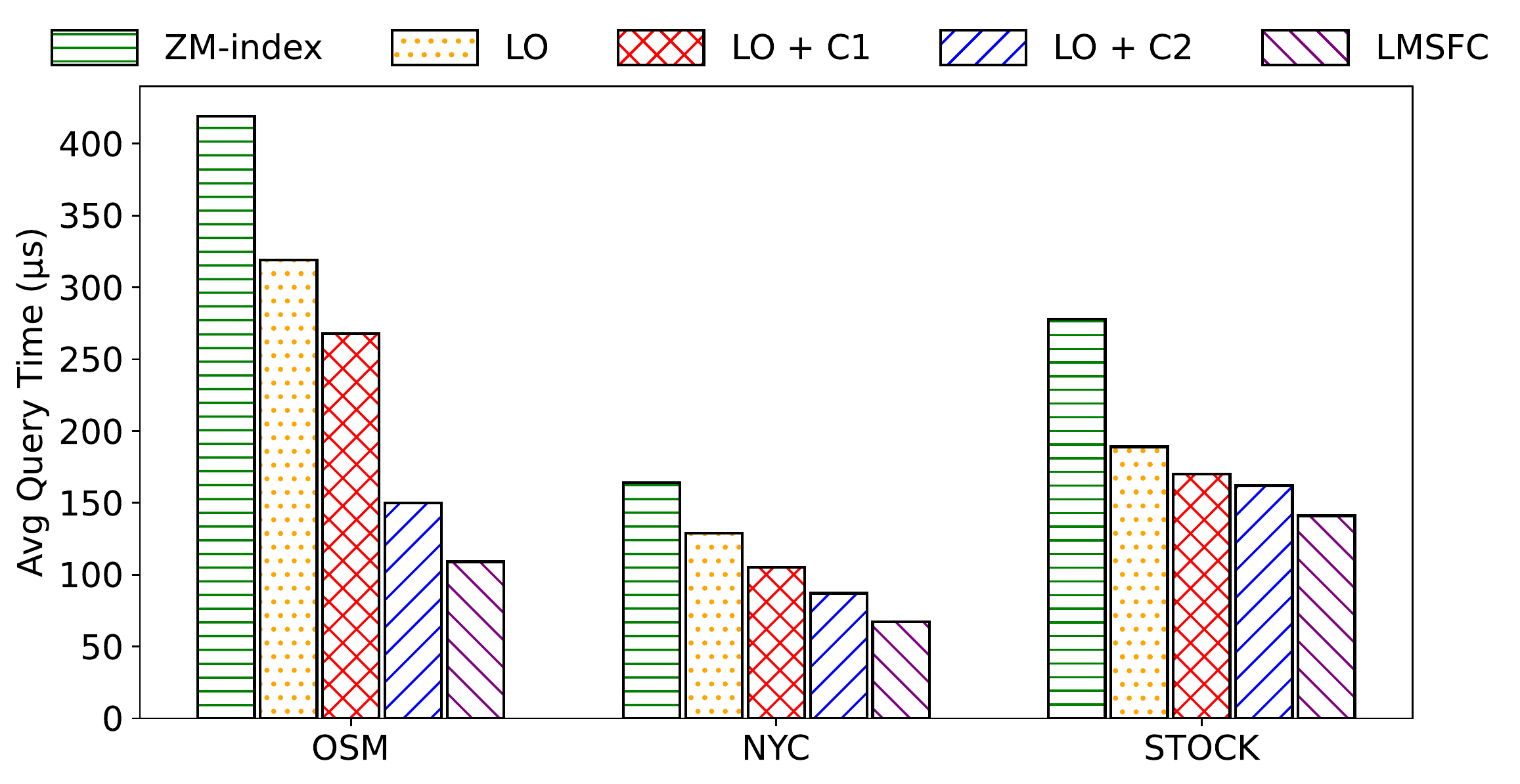}
 
  \caption{Ablation Study}
  \label{fig:ablation}
\end{figure}

\subsection{Ablation Study}

In this section, we investigate the impact of different optimization components
in \MODELNAME on the performance by performing ablation studies. 

We compared the following variants of the proposed method: 
\begin{enumerate}[$\bullet$]
\item \ZM. This baseline uses the fixed \zorder curve with a learned index, with
  fixed-size paging.
\item \textsf{LO}. We replace the \zorder{} in \ZM{} by our learned SFC.
\item \textsf{LO + C1}. On top of \textsf{LO}, we add the sort dimension
  optimization (SD). 
\item \textsf{LO + C2}. On top of \textsf{LO + C1}, we add the Recursive Query
  Splitting (RQS) optimization.
\item \MODELNAME. This is our proposed method, which has Dynamic Programming
  Paging (DP) optimization added to \textsf{LO + C2}. 
\end{enumerate}

Figure~\ref{fig:ablation} illustrates that adding more components can
consistently and continuously improve the baseline model. Across all datasets,
learned \zorder (\textsf{LO}) almost achieves the biggest improvement on \ZM.
This is because learned \zorder has the ability to adapt given query workload
via optimizing.

We notice that \textsf{LO + C1} and \textsf{LO + C2} cannot improve too much on
STOCK dataset since the \md data is too sparse in the higher dimensionality
dataset. As they still use fixed-size paging, points in each page are more
scattered and thus form a much larger MBR. As a result, sort dimension
optimization cannot skip too many points and more pages intersect with the query
window. Note that \MODELNAME{} replaces the fixed-size paging with DP-based
paging, and the above issue is alleviated, hence the noticeable performance
improvement.

\subsection{Query Splitting}

\begin{table}[htbp]
  \small
  \caption{Recursive Query Splitting (\RQS) vs FindNextZaddress (\FNZ)}
  \label{tab:hqp-vs-fnz}
  \centering{}
 
  \begin{tabular}{lrrr}
    \toprule
    Model             & Index Accesses & Avg Query Time ($\mu s$) \\
    \midrule
    \ZM + \RQS        & 18           & 306                  \\
    \ZM + \FNZ        & 1807         & 380                  \\
    \midrule
    \MODELNAME + \RQS & 19           & 109                  \\
    \MODELNAME + \FNZ & 1927         & 206                  \\
    \bottomrule
  \end{tabular}
   
\end{table}

 We investigate the impact of different query splitting strategies. Existing
Z-order-based indexes either do not use any query splitting~\cite{ZM, flood} or
use the query splitting method via repeated invocation of the
$\fn{FindNextZaddress}$ (\FNZ) function first proposed
in~\cite{tropf1981multidimensional} and had been used in
UB-tree~\cite{DBLP:conf/vldb/RamsakMFZEB00}. We name the recursive query
splitting method in our proposal as the \RQS{}.

We experiment with the two splitting strategies on both \ZM{} and \MODELNAME and
show the results in Table~\ref{tab:hqp-vs-fnz}, where ``index accesses'' record
the average number of times the forward index is accessed to perform the
\zaddress{} to page lookup (See explanation in
Section~\ref{sec:query-processing}). %
We can see that our \RQS{} outperforms \FNZ{} and the improvement is especially
significant for \MODELNAME{}. This is mainly because \FNZ{} is invoked for every
page that intersects the query and hence causes great overhead.

  \begin{table}[htbp]
  \small
  \caption{Effect of different $k_{\text{maxsplit}}$}
 
  \label{tab:split}
  \centering{}

  \begin{tabular}{rrrr}
    \toprule
    $k_{\text{maxsplit}}$  & Avg Irrelevant Pages & Avg Query Time ($\mu s$) \\
    \midrule
    0        & 16991          & 150                  \\
    1        & 7531           & 126                  \\
    2        & 4359           & 117                  \\
    3        & 2478           & 112                  \\
    4        & 1288           & 109                  \\
    5        & 523            & 113                  \\
    \bottomrule
  \end{tabular}
   
\end{table}
  We further investigate the effects on different $k_{\text{maxsplit}}$.
  Table~\ref{tab:split} shows we can achieve the best average query performance
  when $k_{\text{maxsplit}} = 4$. Although larger $k_{\text{maxsplit}}$ can
  avoid scanning considerable irrelevant pages, the query performance slightly
  degrades. This is because the query window has been divided into too many
  parts, which significantly increases the overhead on accessing the learned
  index.

\subsection{Paging Methods}

\begin{table}[htpb]
  \small

  \caption{Comparing Different Paging Methods (FP, HP, and DP stands for
    fixed-size paging heuristic paging and dynamic programming paging,
    respectively)}
  \label{tab:PMC}
  \centering
 
  \begin{tabular}{l|rrr}
    \toprule
    Model           & Avg Query Time ($\mu s$) & Index Size (MB) \\\midrule
    \ZM{} + FP     & 419                  &  6.8       \\
    \ZM{} + HP     & 330                  &  8.5       \\
    \ZM{} + DP     & 309                  &  8.8       \\\midrule \MODELNAME +
    FP & 150                  &  6.8       \\
    \MODELNAME + HP & 116                  &  7.4       \\
    \MODELNAME + DP & 109                  &  7.7       \\\bottomrule
  \end{tabular}
 
\end{table}

 We investigate different paging methods on the OSM dataset, and the results are
shown in Table~\ref{tab:PMC}. %
For both \ZM{} and \MODELNAME{}, DP paging shows the best query performance
since DP can minimize our score function, which intuitively corresponds to
relatively densely packed pages.

This helps to reduce the dead space within the page as well as decrease the
probability of overlapping with the queries. Note that HP is only slightly worse
than DP, but typically with much faster packing time (e.g., 35 seconds for HP
versus 546 seconds for DP).

\subsection{Index Learning Process}
\label{exp:training process}

\begin{figure}[htpb]
  \begin{minipage}[t]{0.45\linewidth}
    \centering
    \includegraphics[width=\linewidth]{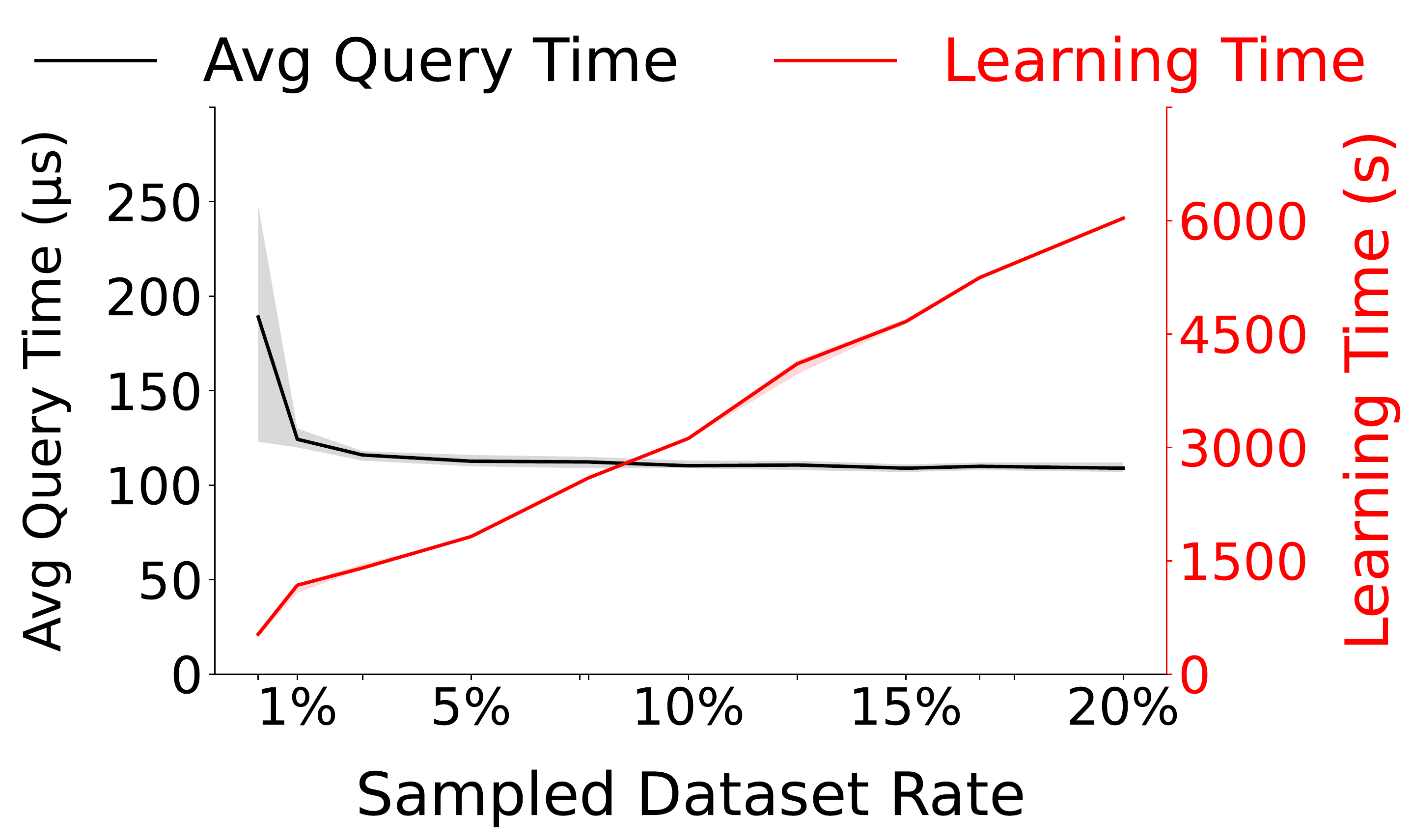}
   
    \caption{Varying Dataset Size}
    \label{fig:learning_processa}
  \end{minipage}
  \hspace{0.05\linewidth}
  \begin{minipage}[t]{0.45\linewidth}
    \centering
    \includegraphics[width=\linewidth]{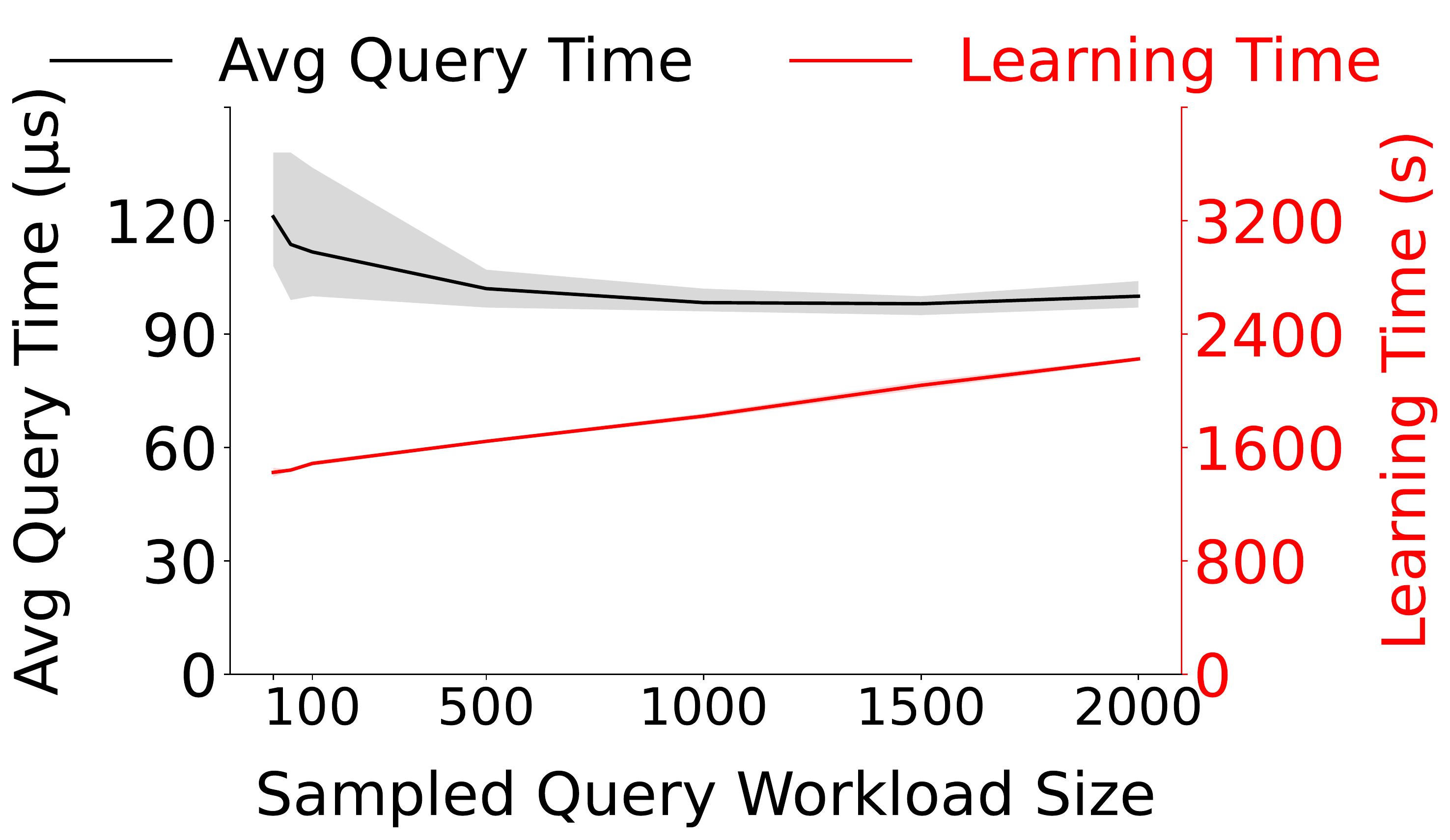}
   
    \caption{Varying Workload Size}
    \label{fig:learning_processb}
  \end{minipage}%
   
\end{figure}

 Learning a good SFC using the entire dataset and query workload is infeasible
due to the prohibitively long time on sorting the dataset according to the given
SFC. Thus, we use sampled datasets and workloads in the learning process, which
can significantly reduce the learning cost without significantly degrading the
query performance. Figure~\ref{fig:learning_processa} and
Figure~\ref{fig:learning_processb} illustrate the learning cost and query
performance on OSM (other datasets show a similar trend) via varying different
sample dataset sizes and query workloads sizes over several trials (minimal and
maximum result is shaded). In Figure~\ref{fig:learning_processa}, when we sample
a small portion of the dataset, the performance has a large variance. Although a
larger sample rate can achieve better performance, the learning process is quite
long. Thus, adopting a 2.5\%-7.5\% sampled rate is good enough to maintain both
fast query time and low learning cost. 
Even if we reduce the sample rate to 0.5\%, we can still achieve better
performance than \Flood{} but incur only half the learning time.

Based on the 5\% sampled dataset, we conduct the experiment on varying query
workload sizes to observe whether a large workload size can achieve better query
performance. The results are displayed in Figure~\ref{fig:learning_processb}.
When a workload size is larger than 500, we can achieve robust performance.

\subsection{Index Size and Index Construction}

\begin{table}[htpb]
  \small
  \caption{Index Size (MB)}
   
  \label{tab:indexsize}
  \centering
    \begin{tabular}{l|rrr}
      \toprule
                  & OSM    & NYC   & STOCK \\\midrule \rstartree & 26.7   & 9.8
      & 8.9   \\
      Flood        & 0.9    & 0.2   & 0.4   \\
      \ZM{}        & 6.8    & 1.6   & 2.6   \\
      \MODELNAME   & 7.7    & 2.0   & 4.4   \\\bottomrule
    \end{tabular}
 
\end{table}

 In Table~\ref{tab:indexsize}, we report the index sizes for the three datasets. 
All the index sizes are small relative to the respective data size. Our proposed
\MODELNAME{} has a larger index size than \ZM{} or \Flood{} partly because we
have optimized page layouts so that pages are not fully filled. Nonetheless, our
index size is still acceptable as it is still significantly smaller than the
traditional \rstartree{}. 

We present the index construction times in Table~\ref{tab:indexcreation}. For
learned indexes, we further distinguish the learning time and the index building
time. 
\rstartree suffers from high index construction time in the large dataset
because its construction requires optimizing some criteria (e.g., dead space,
margin, the overlap between two pages' MBR) for each page. 
\ZM{} has the fastest construction time as there is no learning or optimization
involved. \Flood{} has faster learning and building times than \MODELNAME,
because 
\begin{inparaenum}[(i)]
\item \Flood{}'s model is simpler in that the hyper-parameter space is much
  smaller than ours. In addition, it also optimizes against a learned cost
  model, hence the hyper-parameter search is faster. 
\item Our \MODELNAME{} also includes other optimizations (such as dynamic
  programming-based paging), hence affecting the index building time. 
\end{inparaenum}
Nonetheless, the index construction is done once for a dataset. 

Similar to Flood, if we can collect the training examples from history, we can
train an offline cost model to select the learned SFC with low overhead. We use
history instances as training examples to fit a neural network then we freeze
the parameters of the model. During learning SFC, we can utilize gradient
descent to directly adjust the input to find the optimum. Consequently, the
learning process only takes a few minutes and the performance is competitive
with the proposed BO algorithm.

\begin{table}[tb]
  \small
  \caption{Index Learning and Construction Times (Seconds)}
  \label{tab:indexcreation}
 
  \centering
  \begin{tabular}{l|rrr}
    \toprule
                        & OSM  & NYC  & STOCK \\\midrule \rstartree          &
    9651 & 708  & 864   \\\midrule \ZM{}               & 35   & 5    & 5
    \\\midrule Flood Learning      & 73   & 121  & 431   \\
    Flood Building      & 44   & 6    & 10    \\\midrule \MODELNAME Learning &
    1821 & 672 & 879  \\
    \MODELNAME Building & 546  & 87   & 117   \\
    \bottomrule
  \end{tabular}
\end{table}

\subsection{Handling Data Updates}
\begin{figure}[htpb]
  \begin{minipage}[t]{0.475\linewidth}
    \centering
    \includegraphics[width=\linewidth]{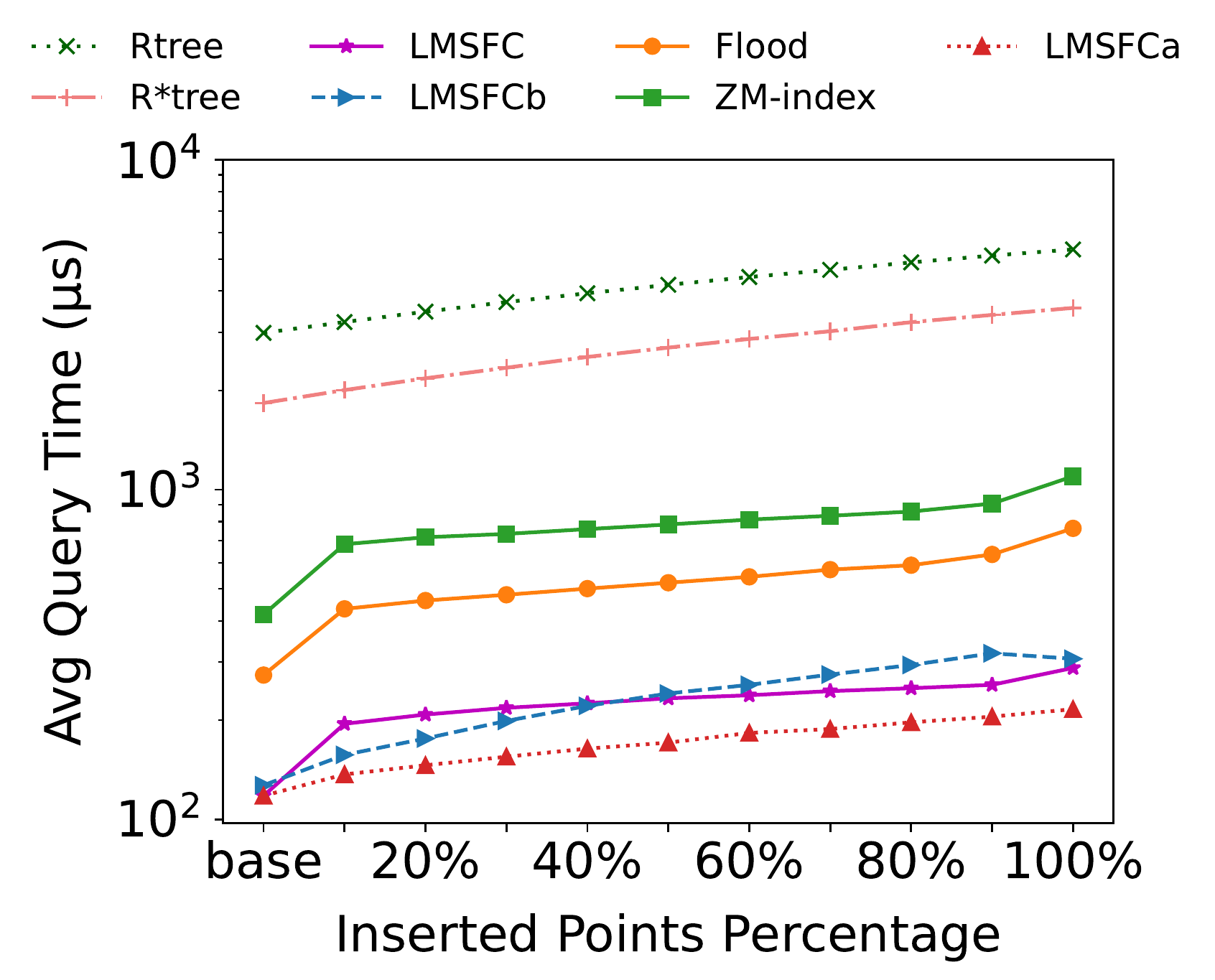}
    \caption{Query Time vs Insert Points Percentage on OSM dataset}
    \label{fig:updatea}
  \end{minipage}
  \hspace{0.05\linewidth}
  \begin{minipage}[t]{0.455\linewidth}
    \centering
    \includegraphics[width=\linewidth]{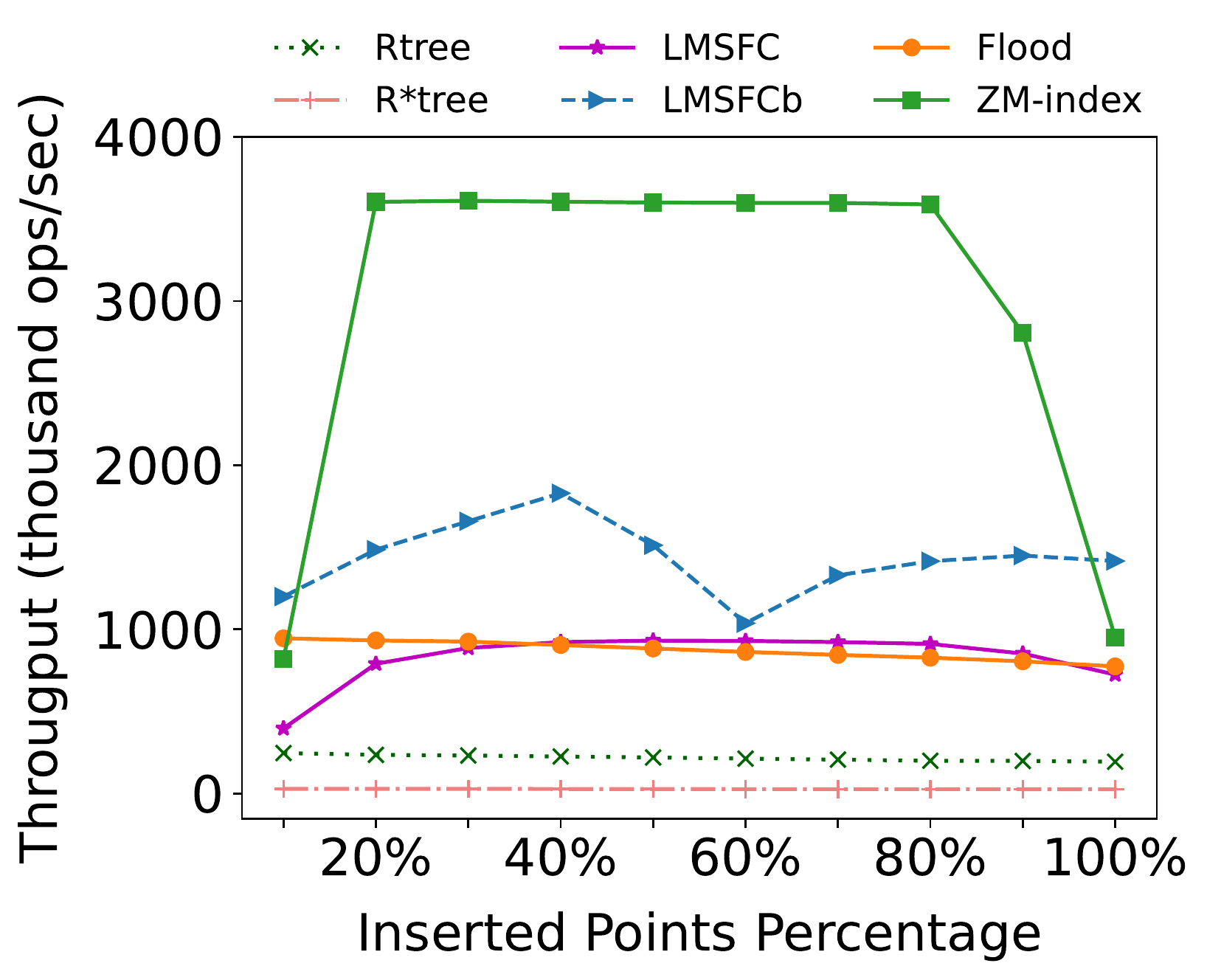}
    \caption{Throughput vs Insert Points Percentage on OSM dataset}
    
    \label{fig:updateb}
  \end{minipage}%
\end{figure}
 Our method can also handle updates by employing an updatable forward index
(e.g., ALEX~\cite{ding2020alex} or LIPP~\cite{DBLP:journals/pvldb/WuZCCWX21}).
In fact, we can employ the traditional $B^+$-tree index as the forward index as
well.

In this section, we adopt ALEX as the forward index since PGM has not an
inherent update method\footnote{PGM adopts LSM-based methods to
handle updates. Although the LSM tree shows good write performance, it performs
worse in range queries because all sorted runs in the LSM tree are
scanned~\cite{DBLP:conf/sigmod/SarkarA22}.} . Deletion can simply
mark a record as "deleted". We only discuss insertion here. Once receiving an
insertion query, LMSFC first locates the data page where the newly inserted
point should be located according to its mapped value. If there is still space,
we insert this point into the page and update the meta-data. Otherwise, we
evenly separate the page into two based on the \zaddress and we also need to
update the meta-data of two pages. We initialize and learn \MODELNAME with the
128M sampled data from the OSM dataset. We compare with two R-tree variants,
\rstartree and R-tree with the linear split method (denoted as R-tree), and we
also implement an updatable Flood and \zm (we also use ALEX to handle data
updates) for comparison. We insert 10\%n to 100\%n data points ($n = 128M$) and
we test the window query performance and report the average query time and
throughput of the insertion query workload. 

The results of query performance after each insertion are shown in
Figure~\ref{fig:updatea}. We observe that updatable learned indexes still show
better query time than traditional index models. LMSFC achieves the best query
performance than other index models. We recognize that while such handling of
updates is feasible, it may deteriorate the quality of our index gradually over
the long run. Thus, after each 10\% data insertion, we trigger a learning
process to find better learned ordering in LMSFC, named LMSFCa. Thus, the
performance can be further improved. Although this incurs extra learning
overhead, the learning phase can happen in a separate instance in a real-world
system. Thus, the learning phase does not block the incoming queries. 

Figure~\ref{fig:updateb} demonstrates high throughput on the learned indexes.
Due to the heuristic split method, R-tree variants show lower throughput. We
observe the low throughput of LMSFC when we insert the first 10\% data. This is
because lots of overflow pages need to be split. We can find a similar trend in
\zm, but \zm shows the highest throughput due to the fact that \zm does not
include any optimization in each page (i.e., page-level sorting). 

Another orthogonal idea that supports the insertion in the learned index is to
maintain a delta index to handle data insertions and delay the merge/split
process ~\cite{DBLP:conf/ppopp/TangWDHWWC20}. We borrow the idea of delayed
update policy and design an LMSFC variant named LMSFCb. In each data page, we
maintain a low-capacity (i.e., half of page size) unsorted delta array to handle
overflow in each page. Once the size of the delta array exceeds the predefined
threshold, we trigger the same splitting method that the updatable LMSFC does.  
The advantage of this method is that the delayed update policy can improve the
insertion throughout while spending a low scan cost of the delta array.
Therefore, Figure~\ref{fig:updatea} and Figure~\ref{fig:updateb} report that
LMSFCb achieves competitive query performance and better insertion throughput.

\section{conclusions and Future Work}
\label{sec:conclusion}

In this paper, we study the problem of learned indexes for \md{} data based on
learned space-filling curves. We devised a framework of learning a special class
of space-filling curves that is amenable to efficient query processing. In
addition, we perform both offline and online optimizations by optimizing the data
placement into pages and query splitting to further improve the query processing
efficiency. Extensive experimental results demonstrate that the proposed method
outperforms both non-learned indexes (such as \rstartree{}) and prior state-of-the-art
learned \md{} indexes (such as \ZM{} and \Flood{}) across a wide range of settings
on three real-world datasets.

Although we focus on learned monotonic SFCs in this paper, our idea and methods 
can be easily generalized to obtain learned non-monotonic SFCs. For example, by dropping the constraints on $\theta$, our method can learn a
non-monotonic SFC. For another example, we can consider other parameterized SFC
families that generalize other well-known SFCs, such as the Hilbert Curve.
We leave such exploration for future work.

\begin{acks}
Wei Wang was supported by HKUST(GZ) Grant G0101000028, GZU-HKUST Joint Research Collaboration Grant GZU22EG04, and Guangzhou Municipal Science and Technology Project (No.\ 2023A0
3J0003). Xin Cao was supported by ARC DP230101534. Jian Gao would like to express his heartfelt gratitude to his girlfriend (Shiman Wu) for her unwavering support throughout the completion of this research paper.
\end{acks}
\bibliographystyle{ACM-Reference-Format}
\bibliography{learned_index}


\begin{thebibliography}{42}


\ifx \showCODEN    \undefined \def \showCODEN     #1{\unskip}     \fi
\ifx \showDOI      \undefined \def \showDOI       #1{#1}\fi
\ifx \showISBNx    \undefined \def \showISBNx     #1{\unskip}     \fi
\ifx \showISBNxiii \undefined \def \showISBNxiii  #1{\unskip}     \fi
\ifx \showISSN     \undefined \def \showISSN      #1{\unskip}     \fi
\ifx \showLCCN     \undefined \def \showLCCN      #1{\unskip}     \fi
\ifx \shownote     \undefined \def \shownote      #1{#1}          \fi
\ifx \showarticletitle \undefined \def \showarticletitle #1{#1}   \fi
\ifx \showURL      \undefined \def \showURL       {\relax}        \fi
\providecommand\bibfield[2]{#2}
\providecommand\bibinfo[2]{#2}
\providecommand\natexlab[1]{#1}
\providecommand\showeprint[2][]{arXiv:#2}

\bibitem[\protect\citeauthoryear{Abdullah-Al-Mamun, Haider, Wang, and Aref}{Abdullah-Al-Mamun et~al\mbox{.}}{2022}]%
        {MDM2022}
\bibfield{author}{\bibinfo{person}{Abdullah-Al-Mamun}, \bibinfo{person}{Ch. Md.~Rakin Haider}, \bibinfo{person}{Jianguo Wang}, {and} \bibinfo{person}{Walid~G. Aref}.} \bibinfo{year}{2022}\natexlab{}.
\newblock \showarticletitle{The ``AI+R''-tree: An Instance-optimized {R-tree}}. In \bibinfo{booktitle}{\emph{MDM}}.
\newblock


\bibitem[\protect\citeauthoryear{Beckmann, Kriegel, Schneider, and Seeger}{Beckmann et~al\mbox{.}}{1990}]%
        {rstartree}
\bibfield{author}{\bibinfo{person}{Norbert Beckmann}, \bibinfo{person}{Hans{-}Peter Kriegel}, \bibinfo{person}{Ralf Schneider}, {and} \bibinfo{person}{Bernhard Seeger}.} \bibinfo{year}{1990}\natexlab{}.
\newblock \showarticletitle{The R*-Tree: An Efficient and Robust Access Method for Points and Rectangles}. In \bibinfo{booktitle}{\emph{Proceedings of the 1990 {ACM} {SIGMOD} International Conference on Management of Data, Atlantic City, NJ, USA, May 23-25, 1990}}. \bibinfo{publisher}{{ACM} Press}, \bibinfo{pages}{322--331}.
\newblock
\urldef\tempurl%
\url{https://doi.org/10.1145/93597.98741}
\showURL{%
\tempurl}


\bibitem[\protect\citeauthoryear{Bentley}{Bentley}{1975}]%
        {kdtree}
\bibfield{author}{\bibinfo{person}{Jon~Louis Bentley}.} \bibinfo{year}{1975}\natexlab{}.
\newblock \showarticletitle{Multidimensional Binary Search Trees Used for Associative Searching}.
\newblock \bibinfo{journal}{\emph{Commun. {ACM}}} \bibinfo{volume}{18}, \bibinfo{number}{9} (\bibinfo{year}{1975}), \bibinfo{pages}{509--517}.
\newblock
\urldef\tempurl%
\url{http://doi.acm.org/10.1145/361002.361007}
\showURL{%
\tempurl}


\bibitem[\protect\citeauthoryear{Davitkova, Milchevski, and Michel}{Davitkova et~al\mbox{.}}{2020}]%
        {ML}
\bibfield{author}{\bibinfo{person}{Angjela Davitkova}, \bibinfo{person}{Evica Milchevski}, {and} \bibinfo{person}{Sebastian Michel}.} \bibinfo{year}{2020}\natexlab{}.
\newblock \showarticletitle{The ML-Index: {A} Multidimensional, Learned Index for Point, Range, and Nearest-Neighbor Queries}. In \bibinfo{booktitle}{\emph{Proceedings of the 23rd International Conference on Extending Database Technology, {EDBT} 2020, Copenhagen, Denmark, March 30 - April 02, 2020}}. \bibinfo{publisher}{OpenProceedings.org}, \bibinfo{pages}{407--410}.
\newblock
\urldef\tempurl%
\url{https://doi.org/10.5441/002/edbt.2020.44}
\showURL{%
\tempurl}


\bibitem[\protect\citeauthoryear{Ding, Minhas, Yu, Wang, Do, Li, Zhang, Chandramouli, Gehrke, Kossmann, Lomet, and Kraska}{Ding et~al\mbox{.}}{2020a}]%
        {ding2020alex}
\bibfield{author}{\bibinfo{person}{Jialin Ding}, \bibinfo{person}{Umar~Farooq Minhas}, \bibinfo{person}{Jia Yu}, \bibinfo{person}{Chi Wang}, \bibinfo{person}{Jaeyoung Do}, \bibinfo{person}{Yinan Li}, \bibinfo{person}{Hantian Zhang}, \bibinfo{person}{Badrish Chandramouli}, \bibinfo{person}{Johannes Gehrke}, \bibinfo{person}{Donald Kossmann}, \bibinfo{person}{David~B. Lomet}, {and} \bibinfo{person}{Tim Kraska}.} \bibinfo{year}{2020}\natexlab{a}.
\newblock \showarticletitle{{ALEX:} An Updatable Adaptive Learned Index}. In \bibinfo{booktitle}{\emph{Proceedings of the 2020 International Conference on Management of Data, {SIGMOD} Conference 2020, online conference [Portland, OR, USA], June 14-19, 2020}}. \bibinfo{publisher}{{ACM}}, \bibinfo{pages}{969--984}.
\newblock
\urldef\tempurl%
\url{https://doi.org/10.1145/3318464.3389711}
\showURL{%
\tempurl}


\bibitem[\protect\citeauthoryear{Ding, Nathan, Alizadeh, and Kraska}{Ding et~al\mbox{.}}{2020b}]%
        {Tsunami}
\bibfield{author}{\bibinfo{person}{Jialin Ding}, \bibinfo{person}{Vikram Nathan}, \bibinfo{person}{Mohammad Alizadeh}, {and} \bibinfo{person}{Tim Kraska}.} \bibinfo{year}{2020}\natexlab{b}.
\newblock \showarticletitle{Tsunami: {A} Learned Multi-dimensional Index for Correlated Data and Skewed Workloads}.
\newblock \bibinfo{journal}{\emph{Proc. {VLDB} Endow.}} \bibinfo{volume}{14}, \bibinfo{number}{2} (\bibinfo{year}{2020}), \bibinfo{pages}{74--86}.
\newblock
\urldef\tempurl%
\url{http://www.vldb.org/pvldb/vol14/p74-ding.pdf}
\showURL{%
\tempurl}


\bibitem[\protect\citeauthoryear{Dong, Indyk, Razenshteyn, and Wagner}{Dong et~al\mbox{.}}{2020}]%
        {DBLP:conf/iclr/DongIRW20}
\bibfield{author}{\bibinfo{person}{Yihe Dong}, \bibinfo{person}{Piotr Indyk}, \bibinfo{person}{Ilya~P. Razenshteyn}, {and} \bibinfo{person}{Tal Wagner}.} \bibinfo{year}{2020}\natexlab{}.
\newblock \showarticletitle{Learning Space Partitions for Nearest Neighbor Search}. In \bibinfo{booktitle}{\emph{8th International Conference on Learning Representations, {ICLR} 2020, Addis Ababa, Ethiopia, April 26-30, 2020}}. \bibinfo{publisher}{OpenReview.net}.
\newblock
\urldef\tempurl%
\url{https://openreview.net/forum?id=rkenmREFDr}
\showURL{%
\tempurl}


\bibitem[\protect\citeauthoryear{Ferragina and Vinciguerra}{Ferragina and Vinciguerra}{2020}]%
        {pgm}
\bibfield{author}{\bibinfo{person}{Paolo Ferragina} {and} \bibinfo{person}{Giorgio Vinciguerra}.} \bibinfo{year}{2020}\natexlab{}.
\newblock \showarticletitle{The PGM-index: a fully-dynamic compressed learned index with provable worst-case bounds}.
\newblock \bibinfo{journal}{\emph{Proc. {VLDB} Endow.}} \bibinfo{volume}{13}, \bibinfo{number}{8} (\bibinfo{year}{2020}), \bibinfo{pages}{1162--1175}.
\newblock
\urldef\tempurl%
\url{http://www.vldb.org/pvldb/vol13/p1162-ferragina.pdf}
\showURL{%
\tempurl}


\bibitem[\protect\citeauthoryear{Finkel and Bentley}{Finkel and Bentley}{1974}]%
        {quadtree}
\bibfield{author}{\bibinfo{person}{Raphael~A Finkel} {and} \bibinfo{person}{Jon~Louis Bentley}.} \bibinfo{year}{1974}\natexlab{}.
\newblock \showarticletitle{Quad trees a data structure for retrieval on composite keys}.
\newblock \bibinfo{journal}{\emph{Acta informatica}} \bibinfo{volume}{4}, \bibinfo{number}{1} (\bibinfo{year}{1974}), \bibinfo{pages}{1--9}.
\newblock


\bibitem[\protect\citeauthoryear{Galakatos, Markovitch, Binnig, Fonseca, and Kraska}{Galakatos et~al\mbox{.}}{2019}]%
        {fittingtree}
\bibfield{author}{\bibinfo{person}{Alex Galakatos}, \bibinfo{person}{Michael Markovitch}, \bibinfo{person}{Carsten Binnig}, \bibinfo{person}{Rodrigo Fonseca}, {and} \bibinfo{person}{Tim Kraska}.} \bibinfo{year}{2019}\natexlab{}.
\newblock \showarticletitle{FITing-Tree: {A} Data-aware Index Structure}. In \bibinfo{booktitle}{\emph{Proceedings of the 2019 International Conference on Management of Data, {SIGMOD} Conference 2019, Amsterdam, The Netherlands, June 30 - July 5, 2019}}. \bibinfo{publisher}{{ACM}}, \bibinfo{pages}{1189--1206}.
\newblock


\bibitem[\protect\citeauthoryear{Gu, Feng, Cong, Long, Wang, and Wang}{Gu et~al\mbox{.}}{[n.d.]}]%
        {DBLP:journals/corr/abs-2103-04541}
\bibfield{author}{\bibinfo{person}{Tu Gu}, \bibinfo{person}{Kaiyu Feng}, \bibinfo{person}{Gao Cong}, \bibinfo{person}{Cheng Long}, \bibinfo{person}{Zheng Wang}, {and} \bibinfo{person}{Sheng Wang}.} \bibinfo{year}{[n.d.]}\natexlab{}.
\newblock \showarticletitle{The RLR-Tree: {A} Reinforcement Learning Based R-Tree for Spatial Data}.
\newblock  (\bibinfo{year}{[n.\,d.]}).
\newblock
\urldef\tempurl%
\url{https://arxiv.org/abs/2103.04541}
\showURL{%
\tempurl}


\bibitem[\protect\citeauthoryear{Guttman}{Guttman}{1984}]%
        {rtree}
\bibfield{author}{\bibinfo{person}{Antonin Guttman}.} \bibinfo{year}{1984}\natexlab{}.
\newblock \showarticletitle{R-Trees: {A} Dynamic Index Structure for Spatial Searching}. In \bibinfo{booktitle}{\emph{SIGMOD'84, Proceedings of Annual Meeting, Boston, Massachusetts, USA, June 18-21, 1984}}. \bibinfo{publisher}{{ACM} Press}, \bibinfo{pages}{47--57}.
\newblock
\urldef\tempurl%
\url{https://doi.org/10.1145/602259.602266}
\showURL{%
\tempurl}


\bibitem[\protect\citeauthoryear{Hadian, Kumar, and Heinis}{Hadian et~al\mbox{.}}{2020}]%
        {ifindex}
\bibfield{author}{\bibinfo{person}{Ali Hadian}, \bibinfo{person}{Ankit Kumar}, {and} \bibinfo{person}{Thomas Heinis}.} \bibinfo{year}{2020}\natexlab{}.
\newblock \showarticletitle{Hands-off Model Integration in Spatial Index Structures} \emph{(\bibinfo{series}{AIDB@VLDB})}.
\newblock


\bibitem[\protect\citeauthoryear{Hilbert and Hilbert}{Hilbert and Hilbert}{1935}]%
        {hilbert1935stetige}
\bibfield{author}{\bibinfo{person}{David Hilbert} {and} \bibinfo{person}{David Hilbert}.} \bibinfo{year}{1935}\natexlab{}.
\newblock \showarticletitle{{\"U}ber die stetige Abbildung einer Linie auf ein Fl{\"a}chenst{\"u}ck}.
\newblock \bibinfo{journal}{\emph{Dritter Band: Analysis{\textperiodcentered} Grundlagen der Mathematik{\textperiodcentered} Physik Verschiedenes: Nebst Einer Lebensgeschichte}} (\bibinfo{year}{1935}), \bibinfo{pages}{1--2}.
\newblock


\bibitem[\protect\citeauthoryear{Hutter, Hoos, and Leyton{-}Brown}{Hutter et~al\mbox{.}}{2011}]%
        {SMBO}
\bibfield{author}{\bibinfo{person}{Frank Hutter}, \bibinfo{person}{Holger~H. Hoos}, {and} \bibinfo{person}{Kevin Leyton{-}Brown}.} \bibinfo{year}{2011}\natexlab{}.
\newblock \showarticletitle{Sequential Model-Based Optimization for General Algorithm Configuration}. In \bibinfo{booktitle}{\emph{Learning and Intelligent Optimization - 5th International Conference, {LION} 5, Rome, Italy, January 17-21, 2011. Selected Papers}} \emph{(\bibinfo{series}{Lecture Notes in Computer Science})}, Vol.~\bibinfo{volume}{6683}. \bibinfo{publisher}{Springer}, \bibinfo{pages}{507--523}.
\newblock
\urldef\tempurl%
\url{https://doi.org/10.1007/978-3-642-25566-3\_40}
\showURL{%
\tempurl}


\bibitem[\protect\citeauthoryear{Kamel and Faloutsos}{Kamel and Faloutsos}{1994}]%
        {kamel1993hilbert}
\bibfield{author}{\bibinfo{person}{Ibrahim Kamel} {and} \bibinfo{person}{Christos Faloutsos}.} \bibinfo{year}{1994}\natexlab{}.
\newblock \showarticletitle{Hilbert R-tree: An Improved R-tree using Fractals}. In \bibinfo{booktitle}{\emph{VLDB'94, Proceedings of 20th International Conference on Very Large Data Bases, September 12-15, 1994, Santiago de Chile, Chile}}. \bibinfo{publisher}{Morgan Kaufmann}, \bibinfo{pages}{500--509}.
\newblock
\urldef\tempurl%
\url{http://www.vldb.org/conf/1994/P500.PDF}
\showURL{%
\tempurl}


\bibitem[\protect\citeauthoryear{Kipf, Marcus, van Renen, Stoian, Kemper, Kraska, and Neumann}{Kipf et~al\mbox{.}}{[n.d.]}]%
        {DBLP:journals/corr/abs-1911-13014}
\bibfield{author}{\bibinfo{person}{Andreas Kipf}, \bibinfo{person}{Ryan Marcus}, \bibinfo{person}{Alexander van Renen}, \bibinfo{person}{Mihail Stoian}, \bibinfo{person}{Alfons Kemper}, \bibinfo{person}{Tim Kraska}, {and} \bibinfo{person}{Thomas Neumann}.} \bibinfo{year}{[n.d.]}\natexlab{}.
\newblock \showarticletitle{{SOSD:} {A} Benchmark for Learned Indexes}.
\newblock  (\bibinfo{year}{[n.\,d.]}).
\newblock
\urldef\tempurl%
\url{http://arxiv.org/abs/1911.13014}
\showURL{%
\tempurl}


\bibitem[\protect\citeauthoryear{Kipf, Marcus, van Renen, Stoian, Kemper, Kraska, and Neumann}{Kipf et~al\mbox{.}}{2020}]%
        {DBLP:conf/sigmod/KipfMRSKK020}
\bibfield{author}{\bibinfo{person}{Andreas Kipf}, \bibinfo{person}{Ryan Marcus}, \bibinfo{person}{Alexander van Renen}, \bibinfo{person}{Mihail Stoian}, \bibinfo{person}{Alfons Kemper}, \bibinfo{person}{Tim Kraska}, {and} \bibinfo{person}{Thomas Neumann}.} \bibinfo{year}{2020}\natexlab{}.
\newblock \showarticletitle{RadixSpline: a single-pass learned index}. In \bibinfo{booktitle}{\emph{Proceedings of the Third International Workshop on Exploiting Artificial Intelligence Techniques for Data Management, aiDM@SIGMOD 2020, Portland, Oregon, USA, June 19, 2020}}. \bibinfo{publisher}{{ACM}}.
\newblock
\urldef\tempurl%
\url{https://doi.org/10.1145/3401071.3401659}
\showDOI{\tempurl}


\bibitem[\protect\citeauthoryear{Kraska, Beutel, Chi, Dean, and Polyzotis}{Kraska et~al\mbox{.}}{2018}]%
        {rmi}
\bibfield{author}{\bibinfo{person}{Tim Kraska}, \bibinfo{person}{Alex Beutel}, \bibinfo{person}{Ed~H. Chi}, \bibinfo{person}{Jeffrey Dean}, {and} \bibinfo{person}{Neoklis Polyzotis}.} \bibinfo{year}{2018}\natexlab{}.
\newblock \showarticletitle{The Case for Learned Index Structures}. In \bibinfo{booktitle}{\emph{Proceedings of the 2018 International Conference on Management of Data, {SIGMOD} Conference 2018, Houston, TX, USA, June 10-15, 2018}}. \bibinfo{publisher}{{ACM}}, \bibinfo{pages}{489--504}.
\newblock
\urldef\tempurl%
\url{https://doi.org/10.1145/3183713.3196909}
\showURL{%
\tempurl}


\bibitem[\protect\citeauthoryear{Lawder and King}{Lawder and King}{2001}]%
        {DBLP:journals/sigmod/LawderK01}
\bibfield{author}{\bibinfo{person}{Jonathan~K. Lawder} {and} \bibinfo{person}{Peter J.~H. King}.} \bibinfo{year}{2001}\natexlab{}.
\newblock \showarticletitle{Querying Multi-dimensional Data Indexed Using the Hilbert Space-filling Curve}.
\newblock \bibinfo{journal}{\emph{{SIGMOD} Rec.}} \bibinfo{volume}{30}, \bibinfo{number}{1} (\bibinfo{year}{2001}), \bibinfo{pages}{19--24}.
\newblock
\urldef\tempurl%
\url{https://doi.org/10.1145/373626.373678}
\showDOI{\tempurl}


\bibitem[\protect\citeauthoryear{Li, Zhang, Sun, Wang, Tsang, and Lin}{Li et~al\mbox{.}}{2020b}]%
        {DBLP:conf/icde/LiZSWT020}
\bibfield{author}{\bibinfo{person}{Mingjie Li}, \bibinfo{person}{Ying Zhang}, \bibinfo{person}{Yifang Sun}, \bibinfo{person}{Wei Wang}, \bibinfo{person}{Ivor~W. Tsang}, {and} \bibinfo{person}{Xuemin Lin}.} \bibinfo{year}{2020}\natexlab{b}.
\newblock \showarticletitle{{I/O} Efficient Approximate Nearest Neighbour Search based on Learned Functions}. In \bibinfo{booktitle}{\emph{36th {IEEE} International Conference on Data Engineering, {ICDE} 2020, Dallas, TX, USA, April 20-24, 2020}}. \bibinfo{publisher}{{IEEE}}, \bibinfo{pages}{289--300}.
\newblock
\urldef\tempurl%
\url{https://doi.org/10.1109/ICDE48307.2020.00032}
\showURL{%
\tempurl}


\bibitem[\protect\citeauthoryear{Li, Lu, Zheng, Yang, and Pan}{Li et~al\mbox{.}}{2020a}]%
        {li2020lisa}
\bibfield{author}{\bibinfo{person}{Pengfei Li}, \bibinfo{person}{Hua Lu}, \bibinfo{person}{Qian Zheng}, \bibinfo{person}{Long Yang}, {and} \bibinfo{person}{Gang Pan}.} \bibinfo{year}{2020}\natexlab{a}.
\newblock \showarticletitle{{LISA:} {A} Learned Index Structure for Spatial Data}. In \bibinfo{booktitle}{\emph{Proceedings of the 2020 International Conference on Management of Data, {SIGMOD} Conference 2020, online conference [Portland, OR, USA], June 14-19, 2020}}. \bibinfo{publisher}{{ACM}}, \bibinfo{pages}{2119--2133}.
\newblock
\urldef\tempurl%
\url{https://doi.org/10.1145/3318464.3389703}
\showURL{%
\tempurl}


\bibitem[\protect\citeauthoryear{Marcus, Kipf, van Renen, Stoian, Misra, Kemper, Neumann, and Kraska}{Marcus et~al\mbox{.}}{2020}]%
        {DBLP:journals/pvldb/MarcusKRSMK0K20}
\bibfield{author}{\bibinfo{person}{Ryan Marcus}, \bibinfo{person}{Andreas Kipf}, \bibinfo{person}{Alexander van Renen}, \bibinfo{person}{Mihail Stoian}, \bibinfo{person}{Sanchit Misra}, \bibinfo{person}{Alfons Kemper}, \bibinfo{person}{Thomas Neumann}, {and} \bibinfo{person}{Tim Kraska}.} \bibinfo{year}{2020}\natexlab{}.
\newblock \showarticletitle{Benchmarking Learned Indexes}.
\newblock \bibinfo{journal}{\emph{Proc. {VLDB} Endow.}} \bibinfo{volume}{14}, \bibinfo{number}{1} (\bibinfo{year}{2020}), \bibinfo{pages}{1--13}.
\newblock
\urldef\tempurl%
\url{https://doi.org/10.14778/3421424.3421425}
\showDOI{\tempurl}


\bibitem[\protect\citeauthoryear{Meagher}{Meagher}{[n.d.]}]%
        {hyperoctree}
\bibfield{author}{\bibinfo{person}{Donald Meagher}.} \bibinfo{year}{[n.d.]}\natexlab{}.
\newblock \showarticletitle{Octree encoding: a new technique for the representation, manipulation and display of arbitrary 3-D objects by computer}.
\newblock \bibinfo{journal}{\emph{Technical Report}} (\bibinfo{year}{[n.\,d.]}).
\newblock


\bibitem[\protect\citeauthoryear{Morton}{Morton}{1966}]%
        {morton}
\bibfield{author}{\bibinfo{person}{Guy~M Morton}.} \bibinfo{year}{1966}\natexlab{}.
\newblock \showarticletitle{A computer oriented geodetic data base and a new technique in file sequencing}.
\newblock  (\bibinfo{year}{1966}).
\newblock


\bibitem[\protect\citeauthoryear{Nathan, Ding, Alizadeh, and Kraska}{Nathan et~al\mbox{.}}{2020}]%
        {flood}
\bibfield{author}{\bibinfo{person}{Vikram Nathan}, \bibinfo{person}{Jialin Ding}, \bibinfo{person}{Mohammad Alizadeh}, {and} \bibinfo{person}{Tim Kraska}.} \bibinfo{year}{2020}\natexlab{}.
\newblock \showarticletitle{Learning Multi-Dimensional Indexes}. In \bibinfo{booktitle}{\emph{Proceedings of the 2020 International Conference on Management of Data, {SIGMOD} Conference 2020, online conference [Portland, OR, USA], June 14-19, 2020}}. \bibinfo{publisher}{{ACM}}, \bibinfo{pages}{985--1000}.
\newblock
\urldef\tempurl%
\url{https://doi.org/10.1145/3318464.3380579}
\showURL{%
\tempurl}


\bibitem[\protect\citeauthoryear{Nievergelt, Hinterberger, and Sevcik}{Nievergelt et~al\mbox{.}}{1984}]%
        {gridfile}
\bibfield{author}{\bibinfo{person}{J{\"{u}}rg Nievergelt}, \bibinfo{person}{Hans Hinterberger}, {and} \bibinfo{person}{Kenneth~C. Sevcik}.} \bibinfo{year}{1984}\natexlab{}.
\newblock \showarticletitle{The Grid File: An Adaptable, Symmetric Multikey File Structure}.
\newblock \bibinfo{journal}{\emph{{ACM} Trans. Database Syst.}} \bibinfo{volume}{9}, \bibinfo{number}{1} (\bibinfo{year}{1984}), \bibinfo{pages}{38--71}.
\newblock
\urldef\tempurl%
\url{https://doi.org/10.1145/348.318586}
\showURL{%
\tempurl}


\bibitem[\protect\citeauthoryear{Nishimura and Yokota}{Nishimura and Yokota}{2017}]%
        {10.1145/3035918.3035934}
\bibfield{author}{\bibinfo{person}{Shoji Nishimura} {and} \bibinfo{person}{Haruo Yokota}.} \bibinfo{year}{2017}\natexlab{}.
\newblock \showarticletitle{QUILTS: Multidimensional Data Partitioning Framework Based on Query-Aware and Skew-Tolerant Space-Filling Curves}. In \bibinfo{booktitle}{\emph{Proceedings of the 2017 ACM International Conference on Management of Data}} \emph{(\bibinfo{series}{SIGMOD '17})}. \bibinfo{publisher}{Association for Computing Machinery}, \bibinfo{pages}{1525–1537}.
\newblock
\showISBNx{9781450341974}
\urldef\tempurl%
\url{https://doi.org/10.1145/3035918.3035934}
\showDOI{\tempurl}


\bibitem[\protect\citeauthoryear{Qi, Liu, Jensen, and Kulik}{Qi et~al\mbox{.}}{2020}]%
        {RSMI}
\bibfield{author}{\bibinfo{person}{Jianzhong Qi}, \bibinfo{person}{Guanli Liu}, \bibinfo{person}{Christian~S. Jensen}, {and} \bibinfo{person}{Lars Kulik}.} \bibinfo{year}{2020}\natexlab{}.
\newblock \showarticletitle{Effectively Learning Spatial Indices}.
\newblock \bibinfo{journal}{\emph{Proc. {VLDB} Endow.}} \bibinfo{volume}{13}, \bibinfo{number}{11} (\bibinfo{year}{2020}), \bibinfo{pages}{2341--2354}.
\newblock
\urldef\tempurl%
\url{http://www.vldb.org/pvldb/vol13/p2341-qi.pdf}
\showURL{%
\tempurl}


\bibitem[\protect\citeauthoryear{Ramsak, Markl, Fenk, Zirkel, Elhardt, and Bayer}{Ramsak et~al\mbox{.}}{2000}]%
        {DBLP:conf/vldb/RamsakMFZEB00}
\bibfield{author}{\bibinfo{person}{Frank Ramsak}, \bibinfo{person}{Volker Markl}, \bibinfo{person}{Robert Fenk}, \bibinfo{person}{Martin Zirkel}, \bibinfo{person}{Klaus Elhardt}, {and} \bibinfo{person}{Rudolf Bayer}.} \bibinfo{year}{2000}\natexlab{}.
\newblock \showarticletitle{Integrating the UB-Tree into a Database System Kernel}. In \bibinfo{booktitle}{\emph{{VLDB} 2000, Proceedings of 26th International Conference on Very Large Data Bases, September 10-14, 2000, Cairo, Egypt}}. \bibinfo{publisher}{Morgan Kaufmann}, \bibinfo{pages}{263--272}.
\newblock


\bibitem[\protect\citeauthoryear{Sachith Gopalakrishna~Pai}{Sachith Gopalakrishna~Pai}{2022}]%
        {IOZ}
\bibfield{author}{\bibinfo{person}{Yanhao~Wang Sachith Gopalakrishna~Pai, Michael~Mathioudakis}.} \bibinfo{year}{2022}\natexlab{}.
\newblock \showarticletitle{Towards an Instance-Optimal Z-Index} \emph{(\bibinfo{series}{AIDB@VLDB})}.
\newblock


\bibitem[\protect\citeauthoryear{Sarkar and Athanassoulis}{Sarkar and Athanassoulis}{2022}]%
        {DBLP:conf/sigmod/SarkarA22}
\bibfield{author}{\bibinfo{person}{Subhadeep Sarkar} {and} \bibinfo{person}{Manos Athanassoulis}.} \bibinfo{year}{2022}\natexlab{}.
\newblock \showarticletitle{Dissecting, Designing, and Optimizing LSM-based Data Stores}. In \bibinfo{booktitle}{\emph{{SIGMOD} '22: International Conference on Management of Data, Philadelphia, PA, USA, June 12 - 17, 2022}}. \bibinfo{publisher}{{ACM}}, \bibinfo{pages}{2489--2497}.
\newblock
\urldef\tempurl%
\url{https://doi.org/10.1145/3514221.3522563}
\showDOI{\tempurl}


\bibitem[\protect\citeauthoryear{Sellis, Roussopoulos, and Faloutsos}{Sellis et~al\mbox{.}}{1987}]%
        {sellis1987r+}
\bibfield{author}{\bibinfo{person}{Timos~K. Sellis}, \bibinfo{person}{Nick Roussopoulos}, {and} \bibinfo{person}{Christos Faloutsos}.} \bibinfo{year}{1987}\natexlab{}.
\newblock \showarticletitle{The R+-Tree: {A} Dynamic Index for Multi-Dimensional Objects}. In \bibinfo{booktitle}{\emph{VLDB'87, Proceedings of 13th International Conference on Very Large Data Bases, September 1-4, 1987, Brighton, England}}. \bibinfo{publisher}{Morgan Kaufmann}, \bibinfo{pages}{507--518}.
\newblock
\urldef\tempurl%
\url{http://www.vldb.org/conf/1987/P507.PDF}
\showURL{%
\tempurl}


\bibitem[\protect\citeauthoryear{Sidlauskas, Chester, Zacharatou, and Ailamaki}{Sidlauskas et~al\mbox{.}}{2018}]%
        {DBLP:conf/icde/SidlauskasCZA18}
\bibfield{author}{\bibinfo{person}{Darius Sidlauskas}, \bibinfo{person}{Sean Chester}, \bibinfo{person}{Eleni~Tzirita Zacharatou}, {and} \bibinfo{person}{Anastasia Ailamaki}.} \bibinfo{year}{2018}\natexlab{}.
\newblock \showarticletitle{Improving Spatial Data Processing by Clipping Minimum Bounding Boxes}. In \bibinfo{booktitle}{\emph{34th {IEEE} International Conference on Data Engineering, {ICDE} 2018, Paris, France, April 16-19, 2018}}. \bibinfo{publisher}{{IEEE} Computer Society}, \bibinfo{pages}{425--436}.
\newblock
\urldef\tempurl%
\url{https://doi.org/10.1109/ICDE.2018.00046}
\showURL{%
\tempurl}


\bibitem[\protect\citeauthoryear{Tang, Wang, Dong, Hu, Wang, Wang, and Chen}{Tang et~al\mbox{.}}{2020}]%
        {DBLP:conf/ppopp/TangWDHWWC20}
\bibfield{author}{\bibinfo{person}{Chuzhe Tang}, \bibinfo{person}{Youyun Wang}, \bibinfo{person}{Zhiyuan Dong}, \bibinfo{person}{Gansen Hu}, \bibinfo{person}{Zhaoguo Wang}, \bibinfo{person}{Minjie Wang}, {and} \bibinfo{person}{Haibo Chen}.} \bibinfo{year}{2020}\natexlab{}.
\newblock \showarticletitle{XIndex: a scalable learned index for multicore data storage}. In \bibinfo{booktitle}{\emph{PPoPP '20: 25th {ACM} {SIGPLAN} Symposium on Principles and Practice of Parallel Programming, San Diego, California, USA, February 22-26, 2020}}. \bibinfo{publisher}{{ACM}}, \bibinfo{pages}{308--320}.
\newblock
\urldef\tempurl%
\url{https://doi.org/10.1145/3332466.3374547}
\showURL{%
\tempurl}


\bibitem[\protect\citeauthoryear{Tian, Yan, Zhao, Huang, and Zhou}{Tian et~al\mbox{.}}{2022}]%
        {DBLP:journals/corr/abs-2204-10028}
\bibfield{author}{\bibinfo{person}{Yao Tian}, \bibinfo{person}{Tingyun Yan}, \bibinfo{person}{Xi Zhao}, \bibinfo{person}{Kai Huang}, {and} \bibinfo{person}{Xiaofang Zhou}.} \bibinfo{year}{2022}\natexlab{}.
\newblock \showarticletitle{A Learned Index for Exact Similarity Search in Metric Spaces}.
\newblock \bibinfo{journal}{\emph{CoRR}}  \bibinfo{volume}{abs/2204.10028} (\bibinfo{year}{2022}).
\newblock
\urldef\tempurl%
\url{https://doi.org/10.48550/arXiv.2204.10028}
\showURL{%
\tempurl}


\bibitem[\protect\citeauthoryear{Tropf and Herzog}{Tropf and Herzog}{1981}]%
        {tropf1981multidimensional}
\bibfield{author}{\bibinfo{person}{Herbert Tropf} {and} \bibinfo{person}{Helmut Herzog}.} \bibinfo{year}{1981}\natexlab{}.
\newblock \showarticletitle{Multidimensional Range Search in Dynamically Balanced Trees}.
\newblock \bibinfo{journal}{\emph{ANGEWANDTE INFO.}} \bibinfo{number}{2} (\bibinfo{year}{1981}), \bibinfo{pages}{71--77}.
\newblock


\bibitem[\protect\citeauthoryear{Wang, Fu, Xu, and Lu}{Wang et~al\mbox{.}}{2019}]%
        {ZM}
\bibfield{author}{\bibinfo{person}{Haixin Wang}, \bibinfo{person}{Xiaoyi Fu}, \bibinfo{person}{Jianliang Xu}, {and} \bibinfo{person}{Hua Lu}.} \bibinfo{year}{2019}\natexlab{}.
\newblock \showarticletitle{Learned Index for Spatial Queries}. In \bibinfo{booktitle}{\emph{20th {IEEE} International Conference on Mobile Data Management, {MDM} 2019, Hong Kong, SAR, China, June 10-13, 2019}}. \bibinfo{publisher}{{IEEE}}, \bibinfo{pages}{569--574}.
\newblock
\urldef\tempurl%
\url{https://doi.org/10.1109/MDM.2019.00121}
\showURL{%
\tempurl}


\bibitem[\protect\citeauthoryear{Wang, Qu, Wu, Wang, and Zhou}{Wang et~al\mbox{.}}{2021}]%
        {learnedcard}
\bibfield{author}{\bibinfo{person}{Xiaoying Wang}, \bibinfo{person}{Changbo Qu}, \bibinfo{person}{Weiyuan Wu}, \bibinfo{person}{Jiannan Wang}, {and} \bibinfo{person}{Qingqing Zhou}.} \bibinfo{year}{2021}\natexlab{}.
\newblock \showarticletitle{Are We Ready For Learned Cardinality Estimation?}
\newblock \bibinfo{journal}{\emph{Proc. {VLDB} Endow.}} \bibinfo{volume}{14}, \bibinfo{number}{9} (\bibinfo{year}{2021}), \bibinfo{pages}{1640--1654}.
\newblock
\urldef\tempurl%
\url{http://www.vldb.org/pvldb/vol14/p1640-wang.pdf}
\showURL{%
\tempurl}


\bibitem[\protect\citeauthoryear{Wu, Zhang, Chen, Chen, Wang, and Xing}{Wu et~al\mbox{.}}{2021}]%
        {DBLP:journals/pvldb/WuZCCWX21}
\bibfield{author}{\bibinfo{person}{Jiacheng Wu}, \bibinfo{person}{Yong Zhang}, \bibinfo{person}{Shimin Chen}, \bibinfo{person}{Yu Chen}, \bibinfo{person}{Jin Wang}, {and} \bibinfo{person}{Chunxiao Xing}.} \bibinfo{year}{2021}\natexlab{}.
\newblock \showarticletitle{Updatable Learned Index with Precise Positions}.
\newblock \bibinfo{journal}{\emph{Proc. {VLDB} Endow.}} \bibinfo{volume}{14}, \bibinfo{number}{8} (\bibinfo{year}{2021}), \bibinfo{pages}{1276--1288}.
\newblock
\urldef\tempurl%
\url{http://www.vldb.org/pvldb/vol14/p1276-wu.pdf}
\showURL{%
\tempurl}


\bibitem[\protect\citeauthoryear{Yang, Chandramouli, Wang, Gehrke, Li, Minhas, Larson, Kossmann, and Acharya}{Yang et~al\mbox{.}}{2020}]%
        {DBLP:conf/sigmod/YangCWGLMLKA20}
\bibfield{author}{\bibinfo{person}{Zongheng Yang}, \bibinfo{person}{Badrish Chandramouli}, \bibinfo{person}{Chi Wang}, \bibinfo{person}{Johannes Gehrke}, \bibinfo{person}{Yinan Li}, \bibinfo{person}{Umar~Farooq Minhas}, \bibinfo{person}{Per{-}{\AA}ke Larson}, \bibinfo{person}{Donald Kossmann}, {and} \bibinfo{person}{Rajeev Acharya}.} \bibinfo{year}{2020}\natexlab{}.
\newblock \showarticletitle{Qd-tree: Learning Data Layouts for Big Data Analytics}. In \bibinfo{booktitle}{\emph{Proceedings of the 2020 International Conference on Management of Data, {SIGMOD} Conference 2020, online conference [Portland, OR, USA], June 14-19, 2020}}. \bibinfo{publisher}{{ACM}}, \bibinfo{pages}{193--208}.
\newblock
\urldef\tempurl%
\url{https://doi.org/10.1145/3318464.3389770}
\showURL{%
\tempurl}


\bibitem[\protect\citeauthoryear{Zhang, Ray, Lu, and Zheng}{Zhang et~al\mbox{.}}{2021}]%
        {DBLP:conf/ssd/ZhangRLZ21}
\bibfield{author}{\bibinfo{person}{Songnian Zhang}, \bibinfo{person}{Suprio Ray}, \bibinfo{person}{Rongxing Lu}, {and} \bibinfo{person}{Yandong Zheng}.} \bibinfo{year}{2021}\natexlab{}.
\newblock \showarticletitle{{SPRIG:} {A} Learned Spatial Index for Range and kNN Queries}. In \bibinfo{booktitle}{\emph{Proceedings of the 17th International Symposium on Spatial and Temporal Databases, {SSTD} 2021, Virtual Event, USA, August 23-25, 2021}}. \bibinfo{publisher}{{ACM}}, \bibinfo{pages}{96--105}.
\newblock
\urldef\tempurl%
\url{https://doi.org/10.1145/3469830.3470892}
\showDOI{\tempurl}


\end{thebibliography}

\end{document}